\documentclass{CSML}

\pdfoutput=1

\usepackage{lastpage}

\def\dOi{13(4:29)2017}
\lmcsheading%
{\dOi}
{1--\pageref{LastPage}}
{}
{}
{Apr.~30, 2016}
{Dec.~22, 2017}
{}

\subjclass{D.3.1 Formal Definitions and Theory, F.3.2 Semantics of Programming Language, F.4.1 Mathematical Logic.}

\usepackage{mathtools}
\usepackage[hidelinks]{hyperref}
\usepackage{multicol}
\theoremstyle{plain}

\usepackage{bussproofs} 
\usepackage{amsmath}
\usepackage{amsfonts}
\usepackage{amssymb}
\usepackage{stmaryrd}
\usepackage[all]{xy}

\usepackage{xcolor}
\usepackage{graphicx}
\usepackage[normalem]{ulem}

\newcommand{\Chole}[1]{\llparenthesis #1\rrparenthesis} 
\newcommand{\ctxC}[1]{\mathtt{C}\Chole{#1}} 
\newcommand{\ctxL}[1]{\mathtt{E}\Chole{#1}} 

\newcommand{\Fv}[1]{\mathsf{fv}(#1)} 
\newcommand{\Nat}{\mathbb{N}} 
\newcommand{\Sem}[2]{\llbracket #1 \rrbracket_{#2}}

\newcommand{\sequence}[2]{\mathop{\lceil #1 \rceil^\mathit{#2}}}

\newcommand{\vg}{\mathsf{v}} 
\newcommand{\sub}[2]{\{{#1}/{#2}\}} 
\newcommand{\msub}[4]{\{{#1}/{#2}, \dots, {#3}/{#4}\}}		
\newcommand{\head}{head}
\newcommand{\Rule}{\mathsf{r}}

\newcommand{\Int}{\scriptscriptstyle\mathit{int}}
\newcommand{\ToInt}{\overset{\Int\ }{\to_\vg}}
\newcommand{\ToIntTrans}{\,\xrightarrow{\!\Int\!}_\vg^{{}_{\scriptstyle{*}}}}
\newcommand{\Head}{\scriptscriptstyle\mathit{h}}
\newcommand{\ToHead}{\overset{\Head\ }{\to_\vg}} 
\newcommand{\ToHeadRule}{\overset{\Head\ }{\to_\Rule}} 
\newcommand{\ToHeadReflTrans}{\xrightarrow{\Head}_\vg^{{}_{\scriptstyle{*}}}} 
\newcommand{\ToHeadRefl}{\xrightarrow{\Head}_\vg^{{}_{\scriptstyle{=}}}} 
\newcommand{\ToHeadBetav}{\overset{\Head\quad}{\to_{\beta_v}}} 
\newcommand{\ToHeadBetavTrans}{\xrightarrow{\Head}_{\beta_v}^{{}_{\scriptscriptstyle{+}}}} 
\newcommand{\ToHeadBetavReflTrans}{\xrightarrow{\Head}_{\beta_v}^{{}_{\scriptstyle{*}}}} 
\newcommand{\ToHeadBetavRefl}{\xrightarrow{\Head}_{\beta_v}^{{}_{\scriptstyle{=}}}} 
\newcommand{\ToHeadSigma}{\overset{\Head\ \ }{\to_{\sigma}}} 
\newcommand{\ToHeadSigmaTrans}{\xrightarrow{\Head}_{\sigma}^{{}_{\scriptscriptstyle{+}}}} 
\newcommand{\ToHeadSigmaReflTrans}{\xrightarrow{\Head}_{\sigma}^{{}_{\scriptstyle{*}}}} 
\newcommand{\ToHeadSigmaRefl}{\xrightarrow{\Head}_{\sigma}^{{}_{\scriptstyle{=}}}} 
\newcommand{\ToHeadSigmaOne}{\overset{\Head\ \ }{\to_{\sigma_1}}} 
\newcommand{\ToHeadSigmaOneReflTrans}{\xrightarrow{\Head}_{\sigma_1}^{{}_{\scriptstyle{*}}}} 
\newcommand{\ToHeadSigmaThree}{\overset{\Head\ \ }{\to_{\sigma_3}}} 
\newcommand{\ToHeadSigmaThreeReflTrans}{\xrightarrow{\Head}_{\sigma_3}^{{}_{\scriptstyle{*}}}} 
\newcommand{\Inter}{\overset{\Int}{\Rightarrow}} 
\newcommand{\ToIntRule}{\,\xrightarrow{\!\Int\!}_\Rule}
\newcommand{\ToIntBetav}{\,\xrightarrow{\!\Int\!}_{\beta_v}}
\newcommand{\ToIntSigmaOne}{\,\xrightarrow{\!\Int\!}_{\sigma_1}}

\newcommand{\InterTrans}{\overset{\Int}{\Rightarrow}^{{}_{\scriptstyle{*}}}} 
\newcommand{\Weak}{\mathsf{w}}
\newcommand{\Strat}{\mathsf{s}}
\newcommand{\ToWeak}{\to_\Weak}
\newcommand{\ToStrat}{\to_\Strat}

\begin{document}

\title[Standardization and conservativity of a refined call-by-value $\lambda$-calculus]{Standardization and conservativity of a refined call-by-value lambda-calculus\rsuper*}

\author[G.~Guerrieri]{Giulio Guerrieri\rsuper a}	
\address{{\lsuper a}Department of Computer Science, University of Oxford, Oxford, United Kingdom}	
\email{\href{mailto:giulio.guerrieri@cs.ox.ac.uk}{giulio.guerrieri@cs.ox.ac.uk}}  
\thanks{This work has been supported by LINTEL TO\textunderscore Call1\textunderscore 2012\textunderscore 0085, a Research Project funded by the ``Compagnia di San Paolo'', and by the A*MIDEX project (ANR-11-IDEX-0001-02) funded by the ``Investissements d'Avenir'' French Government program, managed by the French National Research Agency (ANR)}	

\author[L.~Paolini]{Luca Paolini\rsuper b}	
\address{{\lsuper b}Dipartimento di Informatica, Universit\`a degli Studi di Torino\\
  C.so Svizzera 185, Torino, Italia}	
\email{\href{mailto:paolini@di.unito.it}{paolini@di.unito.it}}  
\author[S.~Ronchi Della Rocca]{Simona Ronchi Della Rocca\rsuper c}
\address{{\lsuper c}Dipartimento di Informatica, Universit\`a degli Studi di Torino\\
  C.so Svizzera 185, Torino, Italia}
\email{\href{mailto:ronchi@di.unito.it}{ronchi@di.unito.it}}


\keywords{call-by-value, standardization, sequentialization, observational equivalence, sigma-reduction, head reduction, parallel reduction, internal reduction, standard sequence,
lambda-calculus, solvability, potential valuability}

\titlecomment{{\lsuper*}This paper is a revised and extended version of \cite{guerrieri2015lipics}, invited for the special issue of TLCA 2015.}


\begin{abstract}
We study an extension of Plotkin's call-by-value lambda-calculus via two commutation rules (sigma-reductions).
These commutation rules are sufficient to remove harmful call-by-value normal forms from the calculus, 
so that it enjoys elegant characterizations of many semantic properties. 
We prove that this extended calculus is a  conservative refinement of Plotkin's one. 
In particular, 
the notions of solvability and potential valuability for this calculus coincide with those for Plotkin's call-by-value lambda-calculus.
The proof rests on a standardization theorem proved by generalizing Takahashi's approach of parallel reductions to our set of reduction rules. 
The standardization is weak (i.e.~redexes are not fully sequentialized) because of overlapping interferences between reductions.  
\end{abstract}

\maketitle


\section{Introduction}
\label{sect:intro}

Call-by-value evaluation is the most common 
parameter passing mechanism for programming languages: 
parameters are evaluated before being passed.
The $\lambda_{v}$-calculus ($\lambda_v$ for short) has been introduced by Plotkin in \cite{Plotkin75} in order to give a formal account of call-by-value evaluation in the context of $\lambda$-calculus.
Plotkin's 
$\lambda_{v}$ has the same term syntax as the ordinary, i.e.~call-by-name, $\lambda$-calculus ($\lambda$ for short), 
but its reduction rule, $\beta_v$, is a restriction of 
$\beta$-reduction for $\lambda$: $\beta_v$-reduction 
reduces a $\beta$-redex only in case the argument is a \emph{value} (i.e.~a variable or an abstraction). 
While $\beta_v$ is enough for evaluation of closed terms not reducing under abstractions, it turned out to be too weak in order to study semantical and operational properties of terms in $\lambda_v$.
This fact makes the theory of $\lambda_v$ (see \cite{EgidiHonsellRonchi92}) 
more complex to be described than that of $\lambda$.
For example, in $\lambda$, $\beta$-reduction is sufficient to characterize solvability and (in addition with $\eta$) separability (see \cite{barendregt84nh} for an extensive survey);
but in order to characterize similar properties for $\lambda_v$,  
only reduction rules incorrect for call-by-value evaluation have been defined (see \cite{PaoliniRonchi99, paolini02ictcs, ronchi04book}): for $\lambda_v$ this is disappointing and requires complex analyses.
The reason of this mismatching is that in $\lambda_v$ there are \emph{stuck $\beta$-redexes} such as $(\lambda y.M)(zz)$, i.e.~$\beta$-redexes that $\beta_v$-reduction will never 
fire because their argument is normal but 
not a value (nor will it ever become one).
The real problem with stuck $\beta$-redexes is that 
they may prevent the creation of other $\beta_v$-redexes, providing \emph{``premature''} \emph{$\beta_v$-normal forms}.
The issue is serious, as it affects termination and thus can impact on the study of 
observational equivalence and other operational properties in $\lambda_v$.
For instance, it is well-known that in $\lambda$ all unsolvable terms are not $\beta$-normalizable (more precisely, solvable terms coincide with the head $\beta$-normalizable ones).
But in $\lambda_v$ (see \cite{ronchi04book,accattoli12lncs,carraro14lncs}) there are 
\emph{unsolvable $\beta_v$-normal} terms,
e.g.~
$M$ and $N$ in Eq.~\ref{eq:premature}:
\begin{align}\label{eq:premature}
  M &= (\lambda y.\Delta)(zz) \Delta & N &= \Delta ((\lambda y.\Delta)(zz)) & \mbox{(where }\Delta &= \lambda x. xx\mbox{)}
\end{align}
Such $M$ and $N$ contain the stuck $\beta$-redex $(\lambda y.\Delta)(zz)$ forbidding evaluation to keep going.
These $\beta_v$-normal forms can be considered ``premature'' because they are unsolvable and so one would expect them to diverge. 
The idea that $M$ and $N$ should behave like the famous divergent term $\Delta\Delta$ is corroborated by the fact that in $\lambda_v$ they are observationally equivalent to $\Delta\Delta$ and have the same semantics as $\Delta\Delta$ in all non-trivial denotational models of $\lambda_v$. 

In a call-by-value setting, the issue of stuck $\beta$-redexes and then of premature $\beta_v$-normal forms arises only when one considers 
\emph{open terms} (in particular, when the reduction under abstractions is allowed, since it forces to deal with ``locally open'' terms). 
Even if to model functional programming languages with a call-by-value parameter passing
, such as OCaml, it is usually enough to just consider closed terms 
and evaluation not reducing under abstractions (i.e.~%
function bodies are evaluated only when all parameters are supplied), the importance to consider open terms in a call-by-value setting can be found, for example, in 
partial evaluation (which evaluates a function when not all parameters are supplied, see \cite{JonesGomardSestoft93}),
in the 
theory of proof assistants such as Coq (in particular, for type checking in a system based on dependent types, see \cite{DBLP:conf/icfp/GregoireL02}), 
or to reason about (denotational or operational) equivalences of terms in $\lambda_v$ that are congruences, or about other theoretical properties of $\lambda_v$ such as separability, potential valuability and solvability, as already mentioned.

\medskip
Here we study the \emph{shuffling calculus} $\lambda_{v}^\sigma$, an extension of $\lambda_{v}$ proposed in \cite{carraro14lncs}. 
It keeps the same term syntax as $\lambda_{v}$ (and $\lambda$) and 
adds to $\beta_v$-reduction two commutation rules, 
$\sigma_{1}$ and $\sigma_{3}$, which ``shuffle'' constructors in order to move stuck $\beta$-redexes and unblock $\beta_v$-redexes that are hidden by the ``hyper-sequential structure'' of terms. 
These commutation rules for $\lambda_v$ (referred also as \emph{$\sigma$-reduction rules}) are similar to Regnier's $\sigma$-rules for $\lambda$ \cite{Regnier92,Regnier94} and inspired by linear logic proof-nets \cite{Girard87}.
It is well-known that $\beta_v$-reduction can be simulated by linear logic cut-elimination via the call-by-value ``boring'' translation $(\cdot)^v$ of $\lambda$-terms into proof-nets \cite[pp.~81-82]{Girard87}, which decomposes the intuitionistic implication as follows: $(A \Rightarrow B)^v = \ !(A^v \multimap B^v)$ (see also \cite{Accattoli15}). 
It turns out that the images under $(\cdot)^v$ of a $\sigma$-redex and its contractum are equal modulo some non-structural cut-elimination steps.
Note that Regnier's $\sigma$-rules are contained in $\beta$-equivalence, while in $\lambda_v$ our $\sigma$-rules are more interesting, as they are not
contained into (i.e.~they enrich) $\beta_v$-equivalence.

One of the benefits of $\lambda_{v}^\sigma$ is that its $\sigma$-rules make all normal forms solvable (indeed $M$ and $N$ in Eq.~\ref{eq:premature} are not normal in $\lambda_{v}^\sigma$). 
More generally, 
$\lambda_v^\sigma$ allows one to characterize semantical and operational properties which are relevant in a call-by-value setting, such as solvability and potential valuability, in an internal and elegant way, as shown in \cite{carraro14lncs}.

\medskip

The main result of this paper is the conservativity of $\lambda_{v}^\sigma$ with respect to $\lambda_{v}$. 
{Namely, 
$\lambda_{v}^\sigma$ is sound with respect to the operational semantics of $\lambda_{v}$ (Corollary~\ref{cor:observational})}, and the notions of potential valuability and solvability characterize, respectively, the same classes of terms in $\lambda_{v}^\sigma$ and $\lambda_{v}$ (Theorem~\ref{thm:valsolv}).
 This fully justifies the project in \cite{carraro14lncs}  where $\lambda^{\sigma}_{v}$ has been introduced as a tool for studying 
{$\lambda_{v}$ by means of reductions sound for $\lambda_{v}$.}
These conservativity results are a consequence of a \emph{standardization} property for $\lambda_v^\sigma$ (Theorem~\ref{thm:standardization}) that formalizes the good interaction arising between $\beta_{v}$-reduction and 
$\sigma$-reduction in $\lambda_{v}^\sigma$.

\medskip

Let us recall the notion of standardization, which has been first  studied in the ordinary \mbox{$\lambda$-calculus} (see 
\cite{CurryFeys58,hindley78,Mitschke79,barendregt84nh}). 
A reduction sequence is \emph{standard} if 
redexes are fired in a given order, and the 
standardization theorem establishes that
every reduction sequence can be transformed into a standard one in a constructive way. 
Standardization is a key tool to grasp the way in which reductions works and sheds some light on 
relationships and dependencies between redexes. It is useful especially to characterize
semantic properties through reduction strategies, 
such as normalization and operational adequacy.

Standardization theorems for $\lambda_{v}$ have been proved by Plotkin \cite{Plotkin75}, Paolini and Ronchi Della Rocca \cite{ronchi04book, paolini04iandc} and Crary \cite{Crary09}. 
Plotkin and Crary define the same notion of standard reduction sequence, based on a partial order between redexes, while Paolini and Ronchi Della Rocca define a different notion, based on a total order between redexes. 
According to the terminology of \cite{klopthesis,lcalcul},  the former gives rise to a weak standardization, while the latter to a strong one.
These standardization theorems for $\lambda_v$ have been proved using a notion of parallel reduction adapted for $\beta_v$-reduction.
Parallel reduction has been originally introduced for $\lambda$ by Tait and Martin-L\"of to prove confluence of $\beta$-reduction: intuitively, it reduces a number of 
$\beta$-redexes in a term simultaneously.
Takahashi \cite{takahashi1989jsl,Takahashi95} has improved this approach and shown that it can be used also to prove standardization for $\lambda$ without involving the tricky notion of residual of a redex, 
unlike the proofs in \cite{CurryFeys58,hindley78,Mitschke79,barendregt84nh}. 
Crary \cite{Crary09} has adapted to $\lambda_v$ Takahashi's me\-thod for standardization.
In order to prove our standardization theorem for $\lambda_v^\sigma$, we extend  the notion of parallel reduction to include all the reductions of $\lambda^{\sigma}_{v}$. 
So, we consider two groups of redexes, $\beta_{v}$-redexes and $\sigma$-redexes (putting together $\sigma_{1}$ and $\sigma_{3}$), and we induce a total order between redexes of the two groups, without imposing any order between $\sigma_{1}$- and $\sigma_{3}$-redexes.
Whenever $\sigma$-redexes are missing, this notion of standardization coincides with that presented in \cite{paolini04iandc,ronchi04book}. 
We show 
it is impossible to strengthen our standardization by (locally) 
giving precedence to $\sigma_{1}$-reduction 
over $\sigma_{3}$-reduction or vice-versa. 
  
As usual, our standardization proof is based on a \emph{sequentialization} result: inner reductions can always be postponed to the head ones, 
according to a non-standard definition of head reduction. 
However, our proof is peculiar 
with respect to 
other ones in the literature.
In particular, our 
parallel reduction does not enjoy the diamond property (we are unaware of interesting parallel reductions that do not enjoy it), 
thus it cannot be used to prove the confluence. This lack is crucially related to the second distinctive aspect of our study, viz.
the presence of 
several kinds of redexes being mutually overlapping (in the sense of \cite{terese2003book}). 

\medskip
The aim of this paper is first of all theoretical: 
to supply a tool for reasoning about semantic and operational properties of Plotkin's $\lambda_{v}$, 
such as observational equivalence, solvability and potential valuability.
The shuffling calculus $\lambda_v^\sigma$ realizes this aim, as shown by the conservativity results with respect to $\lambda_{v}$. 
These results are achieved since $\lambda_v^\sigma$ avoids the problem of 
premature $\beta_v$-normal forms by dealing uniformly with open and closed terms, 
so allowing one to use the classical reasoning by induction on the structure of terms, which is essential in proving semantic and operational  properties. 

In the light of its good behaviour, we believe that $\lambda_v^\sigma$ is also an interesting calculus deserving to be studied in itself and in comparison with other call-by-value extensions of Plotkin's $\lambda_v$ dealing with the problem of stuck $\beta$-redexes, as done for instance in \cite{AccattoliGuerrieri16}.

The approach supplied by 
$\lambda_v^\sigma$ to circumvent the issue of stuck $\beta$-redexes 
might be profitably used also in more practical settings based on a call-by-value evaluation dealing with open terms, 
such as the aforementioned partial evaluation and theory of proof assistants.



%

\subsection*{Related work}
Several variants of $\lambda_{v}$, arising from different perspectives, have been introduced in the literature for modeling the call-by-value computation and dealing with stuck $\beta$-redexes. 
We would like here to mention at least the contributions of Moggi \cite{moggi88ecs,
Moggi89}, Felleisen and Sabry \cite{sabry92lisp,SabryFelleisen93}, Maraist \textit{et al.}~\cite{DBLP:journals/entcs/MaraistOTW95,Maraistetal99}, Sabry and Wadler \cite{SabryWadler97}, Curien and Herbelin \cite{DBLP:conf/icfp/CurienH00}, Dyckhoff and Lengrand \cite{DyckhoffLengrand07}, Herbelin and Zimmerman \cite{herbelin09lncs}, Accattoli and Paolini \cite{accattoli12lncs}, Accattoli and Sacerdoti Coen \cite{AccattoliSacerdoti15}. 
All these proposals are based on the introduction of new constructs to the syntax of $\lambda_{v}$ and/or new reduction rules extending $\beta_v$, so the comparison between them is not easy 
with respect to syntactical properties (some detailed comparison is given~in~\cite{accattoli12lncs,AccattoliGuerrieri16}).
We point out that the calculi introduced in \cite{moggi88ecs,
Moggi89,sabry92lisp,SabryFelleisen93,DBLP:journals/entcs/MaraistOTW95,SabryWadler97,Maraistetal99,DBLP:conf/icfp/CurienH00,
herbelin09lncs} present some variants of our $\sigma_1$ and/or $\sigma_3$ rules, often in a setting with explicit substitutions.
The shuffling calculus $\lambda_v^\sigma$ 
has been introduced by Carraro and Guerrieri in \cite{carraro14lncs} and further studied in \cite{guerrieri2015lipics,guerrieri15wpte,AccattoliGuerrieri16}. 

Regnier \cite{Regnier92,Regnier94} introduced in $\lambda$ the rule $\sigma_1$ (but not $\sigma_3$) and another similar shuffling rule called $\sigma_2$. 
The $\sigma$-rules for $\lambda$ and $\lambda_v$ are different because they are inspired by two different translations of $\lambda$-terms into linear logic proof-nets (see \cite{Girard87}).
A generalization of our and Regnier's $\sigma$-rules is used in \cite{EhrhardGuerrieri16} for a variant of the $\lambda$-calculus subsuming both call-by-name and call-by-value evaluations.

Our approach to prove standardization for $\lambda_v^\sigma$ is inspired by Takahashi's one \cite{takahashi1989jsl,Takahashi95} for $\lambda$ based on parallel reduction, 
adapted for Plotkin's $\lambda_v$ by Crary \cite{Crary09}.

A preliminary version of this paper, focused essentially on the standardization result for $\lambda_v^\sigma$, has been presented in \cite{guerrieri2015lipics}.
 

\subsection*{Outline.}
In Section~\ref{sect:calculus} the 
syntax of $\lambda^{\sigma}_{v}$ with its 
reduction rules is introduced; in Section~\ref{sect:sequentialize} the sequentialization property is proved; Section~\ref{sect:sos} proves the standardization theorem 
for $\lambda_v^\sigma$; in Section \ref{sect:conservative} the main results are given, namely the conservativity of  $\lambda^{\sigma}_{v}$ with respect to Plotkin's $\lambda_v$-calculus, and it is shown that a restricted version of standard sequence supplies a normalizing strategy. 
Section \ref{sect:conclusions} provides some conclusions and hints for future work.

\section{\texorpdfstring{The shuffling calculus: a call-by-value $\lambda$-calculus with $\sigma$-rules}{The shuffling calculus: a call-by-value lambda-calculus with sigma-rules}}
\label{sect:calculus}

In this section we introduce the \emph{shuffling calculus} $\lambda_v^\sigma$, namely the call-by-value $\lambda$-calculus defined in \cite{carraro14lncs} that adds two $\sigma$-reduction rules to the pure (i.e.~without constants) call-by-value $\lambda$-calculus $\lambda_v$ proposed by Plotkin in \cite{Plotkin75}.
The syntax of terms of $\lambda_v^\sigma$ 
is the same as Plotkin's 
$\lambda_v$ and then the same as the ordinary (i.e. call-by-name) $\lambda$-calculus $\lambda$.

\begin{defi}[Term, value]
  Given a countably infinite set $\mathcal{V}$ of \emph{variables} (denoted by
  $x, y, z, \dots$), the sets $\Lambda$ of \emph{terms} and
  $\Lambda_v$ of \emph{values} are defined by mutual induction as follows:
  \begin{equation*}
    \begin{aligned}
      &(\Lambda_{v})\quad&	V, U 		&\Coloneqq \, x \mid \lambda x.M \quad&\text{\emph{values};}\qquad\qquad \ &&
      &(\Lambda)\quad& M, N, L &\Coloneqq \, V \mid MN
      \quad&\text{\emph{terms}.}
    \end{aligned}
  \end{equation*}
\end{defi}

Clearly, $\Lambda_v \subsetneq \Lambda$.
Terms of the form $MN$ (resp.~$\lambda x.M$) are called \emph{applications} (resp.~\emph{abstractions}).
In $\lambda x.M$, the operator $\lambda$ binds its variable $x$ wherever $x$ occurs free in the body $M$.
All terms are considered up to $\alpha$-conversion (i.e.~renaming of bound variables
).

As usual
, $\lambda$'s associate to the right and applications to the left, so $\lambda x y.N$ stands for $\lambda x.(\lambda y.N)$ and $MNL$ for $(MN)L$.
The set of free variables of a term $N$ (i.e.~the set of variables that have occurrences in $N$ not bound by $\lambda$'s) is denoted by $\Fv{N}$: $N$ is \emph{open} if $\Fv{N} \neq \emptyset$, \emph{closed} otherwise.
Given $V_1, \dots, V_n \in \Lambda_v$ and pairwise distinct variables $x_1, \dots, x_n$, $N \msub{V_1}{x_1}{V_n}{x_n}$ denotes the term obtained by the \emph{capture-avoiding simultaneous substitution} of $V_i$ for each free occurrence of $x_i$ in the term $N$ (for all $1 \leq i \leq n$). 
Note that if $N \in \Lambda_v$ then \mbox{$N\{V_1/x_1, \dots, V_n/x_n\} \in \Lambda_v$ (values are closed under substitution).}

\begin{rem}\label{rem:shape}
  Any term 
  can be written in a unique way as $VN_1 \dots N_n$ (a value $V$ recursively applied to $n$ terms $N_1, \dots, N_n$) for some $n \in \Nat$; in particular, values are obtained 
  for $n = 0$.
\end{rem}

From now on, we set $I = \lambda x.x$ and $\Delta = \lambda x. xx$.
One-hole contexts are defined as usual.

\begin{defi}[Context]
  \emph{Contexts} (with exactly one \emph{hole} $\Chole{\cdot}$), denoted by
  $\mathtt{C}$, are defined 
  via the grammar:
  \begin{equation*}
    \mathtt{C} \, \Coloneqq \, \Chole{\cdot} \, \mid \, \lambda x.\mathtt{C} \, \mid \, \mathtt{C} M \, \mid \, M \mathtt{C} \,.
  \end{equation*}
  
  Let $\mathtt{C}$ be a context.
  The set of free variables of $\mathtt{C}$ is denoted by $\Fv{\mathtt{C}}$.
  We use $\mathtt{C}\Chole{M}$ for the term obtained by the
  capture-allowing substitution of the term $M$ for the hole
  $\Chole{\cdot}$ in 
  $\mathtt{C}$.
\end{defi}

The set of 
$\lambda_v^\sigma$-reduction rules contains Plotkin's $\beta_v$-reduction rule together with two simple commutation rules called $\sigma_1$ and $\sigma_3$, studied in \cite{carraro14lncs}.

\begin{defi}[Reduction rules]\label{def:lsv-reduction}
For any $M, N, L \in \Lambda$ and any $V \in \Lambda_v$, we define the following binary relations on $\Lambda$:
\begin{equation*}
\begin{aligned}
  (\lambda x.M)V    		&\mapsto_{\beta_v}    M\{V/x\} 		
\\
(\lambda x.M)NL   		&\mapsto_{\sigma_1}   (\lambda x.ML)N 	&&\text{ with } x\notin \Fv{L} \\
 V((\lambda x.L)N) 		&\mapsto_{\sigma_3} (\lambda x.VL)N 	&&\text{ with } x\notin \Fv{V}.
\end{aligned}
\end{equation*}


We set $\mapsto_{\sigma} \, = \ \mapsto_{\sigma_1} \!\cup \mapsto_{\sigma_3}$ and $\mapsto_{\mathsf{v}} \, = \ \mapsto_{\beta_v} \!\cup \mapsto_{\sigma}$.

For any $\mathsf{r} \in \{\beta_v, \sigma_1, \sigma_3, \sigma, \vg\}$, if $M \mapsto_\mathsf{r} M'$ then $M$ is a \emph{$\mathsf{r}$-redex} and $M'$ is its \emph{$\mathsf{r}$-contractum}. 
In the same sense, a term of the shape $(\lambda x.M)N$ (for any $M, N \in \Lambda$) is a \emph{$\beta$-redex}.
\end{defi}

The side conditions for $\mapsto_{\sigma_1}$ and $ \mapsto_{\sigma_3}$ in Definition~\ref{def:lsv-reduction} can be always fulfilled by $\alpha$-renaming. 
Clearly, any $\beta_v$-redex is a $\beta$-redex but the converse does not hold: $(\lambda x. z) (yI)$ is a $\beta$-redex but not a $\beta_v$-redex. 
Redexes of different kind may \emph{overlap} (in the sense of \cite{terese2003book}
):
e.g. the term $\Delta I \Delta$ is a $\sigma_1$-redex and contains the $\beta_v$-redex $\Delta I$; the term $\Delta (I \Delta) (xI)$ is a $\sigma_1$-redex and contains the $\sigma_3$-redex $\Delta(I\Delta)$, which contains in turn the $\beta_v$-redex $I \Delta$.

\begin{rem}\label{rmk:alternative-sigma}
  The relation $\mapsto_\sigma$ can be defined as a unique reduction rule, namely
$$\ctxL{(\lambda x.M)N} \mapsto_{\sigma} (\lambda x.\ctxL{M})N$$
 where $\mathtt{E}$ is a context of the form $\Chole{\cdot}L$ or $V\Chole{\cdot}$ (for any $L \in \Lambda$ and $V \in \Lambda_v$) such that $x \notin \Fv{\mathtt{E}}$.
\end{rem}

Let $\mathsf{R}$ be a binary relation on $\Lambda$. 
We denote by $\mathsf{R}^*$ (resp.\ $\mathsf{R}^+$; $\mathsf{R}^=$) its reflexive-transitive (resp.\ transitive; reflexive) closure.

\begin{defi}[Rewriting notations and terminology]
\label{def:rewriting}
  Let 
  $\Rule \in \{\beta_v, \sigma_1, \sigma_3, \sigma, \vg\}$.
  \begin{itemize}
  \item The \emph{$\Rule$-reduction} $\to_\Rule$ is the
    contextual closure of $\mapsto_\Rule$,
    i.e.\ 
    $M \to_\Rule M'$ iff there is a context $\mathtt{C}$ and $N,
    N' \in \Lambda$ such that $M = \ctxC{N}$, $M' = \ctxC{N'}$ and $N
    \mapsto_\Rule N'$.

  \item  The \emph{$\Rule$-equivalence} $=_\Rule$ is the
    congruence relation on $\Lambda$ generated by $\mapsto_\Rule$, i.e.\ the
    reflexive-transitive and symmetric closure of $\to_\Rule$.

\item    Let $M$ be a term: $M$ is \emph{$\Rule$-normal} if there is
    no term $N$ such that $M \to_\Rule N$; $M$ is
    \emph{$\mathsf{r}$-normalizable} if there is a $\mathsf{r}$-normal
    term $N$ such that $M \to_\Rule^* N$, and we then say that $N$ is a \emph{$\Rule$-normal form of $M$}; $M$ is \emph{strongly
      $\Rule$-normalizable} if it does not exist an infinite sequence of $\Rule$-reductions starting from $M$.
    Finally, $\to_\Rule$ is \emph{strongly
      normalizing} if every $N \in \Lambda$ is strongly $\Rule$-normalizable.
  \end{itemize}
\end{defi}

\noindent From Definitions~\ref{def:lsv-reduction} and \ref{def:rewriting}, it follows immediately that $\to_{\vg} \, = \, \to_{\beta_v} \!\cup \to_{\sigma}$ with $\to_{\sigma} \, \subsetneq \, \to_{\vg}$ and $\to_{\beta_v} \, \subsetneq \, \to_{\vg}$, and also that $\to_{\sigma} \, = \, \to_{\sigma_1} \!\cup \to_{\sigma_3}$ with $\to_{\sigma_1} \, \subsetneq \, \to_{\sigma}$ and $\to_{\sigma_3} \, \subsetneq \, \to_{\sigma}$.

\begin{rem}\label{rmk:fromvalue} 
  Given $\Rule \in \{\beta_v, \sigma_1, \sigma_3, \sigma, \vg\}$ (resp.\ $\Rule \in \{\sigma_1, \sigma_3, \sigma\}$), values are closed under $\Rule$-re\-du\-ction (resp.\ $\mathsf{r}$-expansion): for any $V \in \Lambda_v$, if $V \to_{\Rule} M$ (resp.\ $M \to_\Rule V$) then $M \in \Lambda_v$ and
  more precisely $V = \lambda x.N$ and $M = \lambda x.L$ for some $N, L \in \Lambda$ with $N \to_\Rule L$ (resp.~$L \to_\Rule N$).
\end{rem}

\begin{prop}[Basic properties of 
reductions, \cite{Plotkin75,carraro14lncs}]\label{prop:general-properties}
  The $\sigma$-reduction is confluent and strongly normalizing. 
  The $\beta_v$- and $\vg$-reductions are confluent.
\end{prop}
\begin{proof}
  Confluence of $\beta_v$-reduction has been proved in \cite{Plotkin75}.
  The  $\sigma$-reduction is strongly confluent in the sense of \cite{huet1980acm},
  whence confluence of $\sigma$-reduction follows.
  The $\vg$-reduction is not strongly confluent
and a more sophisticated proof is needed. All details (as well as the proof that $\sigma$-reduction is strongly normalizing) are  in \cite{carraro14lncs}.
 \end{proof}
 
By confluence (Proposition~\ref{prop:general-properties}), for any $\Rule \in \{\beta_v, \sigma, \vg\}$ we have that: $M =_{\Rule} N$ iff $M \to_\Rule^* L \,\,{}_\Rule^*\!\!\!\leftarrow N$ for some term $L$; and any $\Rule$-normalizable term has a \emph{unique} $\Rule$-normal form.

Looking at the tree-like representation of terms, there is a clear symmetry between the commutation rules $\mapsto_{\sigma_1}$ and $\mapsto_{\sigma_3}$ in Definition~\ref{def:lsv-reduction}, except for the fact that $\mapsto_{\sigma_3}$ requires that the ``shuffled'' term is a value.
If in the definition of $\mapsto_{\sigma_3}$ the value $V$ were replaced by any term,  $\to_\sigma$ and $\to_\vg$ would not be (locally) confluent: 
consider all the reduction sequences obtained from $(\lambda x.x')(xI) ((\lambda y.y')(yI))$ using $\to_{\beta_v}$, $\to_{\sigma_1}$ and the unrestricted version of $\to_{\sigma_3}$.

The \emph{shuffling calculus} or \emph{$\lambda_v^\sigma$-calculus} ($\lambda_v^\sigma$ for short) is the set $\Lambda$ of terms endowed with the reduction $\to_{\vg}$. 
The set $\Lambda$ endowed with the reduction $\to_{\beta_v}$ is the $\lambda_v$-calculus ($\lambda_v$ for short), i.e.\ Plotkin's pure call-by-value $\lambda$-calculus \cite{Plotkin75}, a sub-calculus of $\lambda_v^{\sigma}$.

\begin{exa}\label{ex:reductions}
  Recalling the terms $M$ and $N$ in Eq.~\ref{eq:premature}, one has that
  $M \!=\! (\lambda y.\Delta)(xI) \Delta \!\to_{\sigma_1}\! (\lambda y.\Delta\Delta) (x I) \!\to_{\beta_v}\! (\lambda y.\Delta\Delta) (x I) \!\to_{\beta_v} \!\dots$ and $N \!=\! \Delta((\lambda y.\Delta)(x I)) \!\to_{\sigma_3}\! (\lambda y.\Delta\Delta) (x I) \!\to_{\beta_v}\! (\lambda y.\Delta\Delta) (x I) \!\to_{\beta_v} \!\dots$ are the only possible $\mathsf{v}$-reduction paths from $M$ and $N$ respectively: 
  $M$ and $N$ are not $\vg$-normalizable and $M =_\vg N$.  
  But $M$ and $N$ are $\beta_v$-normal ($(\lambda y.\Delta) (x I)$ is a stuck $\beta$-redex) 
  and different, hence $M \neq_{\beta_v} N$ by confluence of $\to_{\beta_v}$ (Proposition~\ref{prop:general-properties}).
\end{exa}

Example~\ref{ex:reductions} shows how $\sigma$-reduction shuffles constructors and moves stuck $\beta$-redex in order to unblock $\beta_v$-redexes which are hidden by the ``hyper-sequential structure'' of terms, avoiding ``premature'' normal forms.
An alternative approach to circumvent the issue of stuck $\beta$-redexes is given by $\lambda_\mathsf{vsub}$, the call-by-value $\lambda$-calculus with explicit substitutions introduced in \cite{accattoli12lncs}, 
where hidden $\beta_v$-redexes are reduced using rules acting at a distance. 
In \cite{AccattoliGuerrieri16} it has been shown that $\lambda_\mathsf{vsub}$ and $\lambda_v^\sigma$ can be 
embedded in each other preserving termination and divergence.
Interestingly, both calculi are inspired by an analysis of Girard's ``boring'' call-by-value translation of $\lambda$-terms into linear logic proof-nets \cite{Girard87,Accattoli15}.

\section{\texorpdfstring{Sequentialization}{Sequentialization}}\label{sect:sequentialize}

Standardization is a consequence of a \emph{sequentialization} property: every $\vg$-reduction sequence can always be rearranged so 
that head $\vg$-reduction steps precede internal ones. 
To prove this 
sequentialization (Theorem~\ref{thm:trifurcate}), we adapt to $\lambda_v^\sigma$ Takahashi's method \cite{Takahashi95,Crary09} based on parallel reduction. 
This is the most technical part of the paper: for the sake of readability, this proof 
{together with all needed lemmas 
are collected in S
ection~\ref{subsect:proof}.}

{First, we partition} $\vg$-reduction into head $\vg$-reduction and internal $\vg$-reduction. 
In turn, head $\vg$-reduction 
divides up into head $\beta_v$-reduction and head $\sigma$-reduction.
Their definitions are driven by the shape of terms, as given in 
Remark~\ref{rem:shape}.
  
\begin{defi}[Head $\beta_v$-reduction]\label{def:headbetav}
  The \emph{head $\beta_v$-reduction} $\ToHeadBetav \, \subseteq \Lambda \times \Lambda$ is defined inductively by the following rules (where $m \in \Nat$):
  \begin{center}
  {\small
    \def\ScoreOverhang{0pt}
    \AxiomC{}
    \RightLabel{\scriptsize $\beta_v$}
    \UnaryInfC{$(\lambda x . M)V M_1\dots M_m \ToHeadBetav M\sub{V}{x}M_1\dots M_m$}
    \DisplayProof
    \quad
    \def\ScoreOverhang{0pt}
    \AxiomC{$N \ToHeadBetav N'$}
    \RightLabel{\scriptsize $\mathit{right}$}
    \UnaryInfC{$VN M_1 \dots M_m \ToHeadBetav VN' M_1 \dots M_m$}
    \DisplayProof}.
  \end{center}
\end{defi}

Head $\beta_v$-reduction 
is the reduction strategy choosing at every step the (unique, if any) leftmost-outermost $\beta_v$-redex not in the scope of a $\lambda$
: thus, it is a \emph{deterministic} reduction (i.e.~
a partial function from $\Lambda$ to $\Lambda$) and does not reduce values.
It coincides with the 
``left reduction'' defined in \cite[p.\,136]{Plotkin75} for $\lambda_v$,
called ``evaluation'' in \cite{SabryFelleisen93,
Lassen05,Crary09}, and it models call-by-value evaluation as implemented in functional programming languages such as OCaml.
Head $\beta_v$-reduction 
is often equivalently defined either by using the rules 
  \begin{center}
  {\small
    \def\ScoreOverhang{0pt}
    \AxiomC{}
    \UnaryInfC{$(\lambda x . M)V \ToHeadBetav M\sub{V}{x}$}
    \DisplayProof
    \qquad\quad
    \def\ScoreOverhang{0pt}
    \AxiomC{$N \ToHeadBetav N'$}
    \UnaryInfC{$VN \ToHeadBetav VN'$}
    \DisplayProof
    \qquad\quad
    \def\ScoreOverhang{0pt}
    \AxiomC{$M \ToHeadBetav M'$}
    \UnaryInfC{$MN \ToHeadBetav M'N$}
    \DisplayProof
  }
  \end{center}
or as the closure of the relation $\mapsto_{\beta_v}$  under evaluation contexts $\mathtt{E} \Coloneqq \Chole{\cdot} \, \mid \, \mathtt{E} M \, \mid \, V \mathtt{E}$. 
We prefer our presentation since it allows more concise proofs and stresses in a more explicit way how head $\beta_v$-reduction acts on the general shape of terms, 
as given in Remark~\ref{rem:shape}.

\begin{defi}[Head $\sigma$-reduction]\label{def:headsigma}
  The \emph{head $\sigma$-reduction} $\ToHeadSigma\, \subseteq \Lambda \times \Lambda$ is defined inductively by the following rules 
  (where $m \in \Nat$, and $x \notin \Fv{L}$ in the rule $\sigma_1$, $x \notin \Fv{V}$ in the rule~$\sigma_3$):
  \begin{center}
  {\small
    \def\ScoreOverhang{0pt}
    \AxiomC{}
    \RightLabel{\scriptsize $\sigma_1$}
    \UnaryInfC{$(\lambda x.M)NL M_1\dots M_m \ToHeadSigma (\lambda x.ML)N M_1\dots M_m$}
    \DisplayProof
    \quad
    \def\ScoreOverhang{0pt}
    \AxiomC{$N \ToHeadSigma N'$}
    \RightLabel{\scriptsize $\mathit{right}$}
    \UnaryInfC{$VN M_1 \dots M_m \ToHeadSigma  VN' M_1 \dots M_m$}
    \DisplayProof

    \vspace{.6\baselineskip}
    \AxiomC{}
    \RightLabel{\scriptsize $\sigma_3$}
    \UnaryInfC{$V((\lambda x.L)N) M_1\dots M_m  \ToHeadSigma (\lambda x.VL)N M_1\dots M_m$}
    \DisplayProof.
  }
  \end{center}  

\noindent The \emph{head $\sigma_1$-}(resp.~\emph{head $\sigma_3$-})\emph{reduction} is $\ToHeadSigmaOne \,=\, \to_{\sigma_1} \cap \ToHeadSigma$ (resp.~$\ToHeadSigmaThree \,=\, \to_{\sigma_3} \cap \ToHeadSigma$).
\end{defi}

Head $\sigma$-reduction 
is a non-deterministic reduction,
since it reduces at every step 
``one of the 
leftmost-outermost'' $\sigma_1$- or $\sigma_3$-redexes not in the scope of a $\lambda$
: such 
head 
$\sigma$-redexes may 
be not unique and overlap, 
e.g.~the term 
$N$ in Figure~\ref{fig:sigma-overlap} 
is a head $\sigma_1$-redex containing 
the head $\sigma_3$-redex $(\lambda y.y')(\Delta(xI))
$
, or the term 
$I(\Delta(I(xI)))$ 
is a head $\sigma_3$-redex containing 
another head $\sigma_3$-redex $\Delta(I(xI))$ (see Definition~\ref{def:headvi} below for the \mbox{formal definition of head redex).}

\begin{defi}[Head $\vg$-reduction, internal $\vg$-reduction, head redex]\label{def:headvi}
The \emph{head} $\vg$-\emph{reduction} is $\ToHead \, = \, \ToHeadBetav \!\cup\! \ToHeadSigma$.
The \emph{internal} $\vg$-\emph{reduction} is $\ToInt \, = \, \to_\vg \!\smallsetminus\! \ToHead$.

Given $\Rule \in \{\beta_v, \sigma_1, \sigma_3, \sigma, \vg\}$
, a \emph{head $\Rule$-redex} \emph{of a term $M$} is a $\Rule$-redex $R$ occurring in $M$ such that $M \ToHeadRule N$ for some 
  term $N$ obtained from $M$ by replacing $R$ with its $\Rule$-contractum.
\end{defi} 

Note that $\mapsto_{\beta_v} \, \subsetneq \, \ToHeadBetav \, \subsetneq \, \to_{\beta_v}$ and $\mapsto_{\sigma} \, \subsetneq \, \ToHeadSigma \, \subsetneq \, \to_{\sigma}$ and $\mapsto_\vg \, \subsetneq \, \ToHead \, \subsetneq \, \to_\vg$.
It is immediate to check that, for any $M,M',N,L \in \Lambda$, $M \to_\vg M'$ implies $NLM \ToInt NLM'$.

Head $\vg$-reduction is non-deterministic since head $\sigma$-reduction is so, 
and since head $\beta_v$- and head $\sigma_1$-(resp.~$\sigma_3$-)redexes may overlap, as in the term $I \Delta I$ (resp.~$I (\Delta I)$). 
Also, Figure~\ref{fig:sigma-overlap} shows that head $\sigma$- and head $\vg$-reductions 
are not (locally) confluent and a term may have several 
head $\sigma$/$\vg$-normal forms, indeed 
$N'_0$ and $N'_1$ are 
head $\sigma$/$\vg$-normal forms of $N$ but $N'_0 \neq N'_1$. 
However, this does not contradict the confluence of $\sigma$- and $\vg$-reductions because $N'_1 \to_\sigma N'_0$ 
by performing an internal $\vg$-reduction step.
Also, Remark~\ref{rmk:unique} in Section~\ref{sect:conservative} states that if a term head $\vg$-reduces to a value $V$%
, then $V$ is its unique head $\vg$-normal form
.

\medskip
Now we can state the first main result of this paper, namely the sequentialization theorem (Theorem~\ref{thm:trifurcate}), saying that 
any $\vg$-reduction sequence can be sequentialized into a head $\beta_v$-reduction sequence followed by a head $\sigma$-reduction sequence, followed by an internal $\vg$-reduction sequence. In ordinary $\lambda$-calculus, the well-known result corresponding to Theorem~\ref{thm:trifurcate} says that a $\beta$-reduction sequence can be factorized in a head $\beta$-reduction sequence followed by an internal $\beta$-reduction sequence (see for example \cite[Corollary~2.6]{Takahashi95}).


\begin{thm}[Sequentialization; its proof is in Section~\ref{subsect:proof}]
\label{thm:trifurcate}\hfill
  If $M \to_{\vg}^* M'$ then there exist $L,N \in \Lambda$ such that   $M \ToHeadBetavReflTrans L \ToHeadSigmaReflTrans N \ToIntTrans M'$.
\end{thm}
%
%
%
\begin{figure}
\scalebox{0.85}{
  \framebox{
  $$\xymatrix@C=+0cm@R=+0.5cm@M=3mm@L=.8mm{ 
  & \hspace{-5mm} N=(\lambda y.y')(\Delta (xI))I\ar@{->}@(dl,u)[dl]_{h}^>{\sigma_1} \ar@{->}@(dr,u)[dr]^h_>{\sigma_3} \hspace{-5mm}& \\
  { N}_0= (\lambda y.y'I)(\Delta (xI)) \ar@{->}[d]_{h}^>{\sigma_3} & & (\lambda z.(\lambda y.y')(zz))(xI)I={ N}_1  \ar@{->}^h[d]_>{\sigma_1}  \\
  { N}'_0=(\lambda z.(\lambda y.y'\!I)(zz))(xI) \;& & (\lambda z.(\lambda y.y')(zz)I)(xI)={ N}'_1\ar@{.>}[ll]^>{\sigma_1} }$$}
}
\caption{Overlapping of 
(head) $\sigma$-redexes
.}
\label{fig:sigma-overlap}
\end{figure}
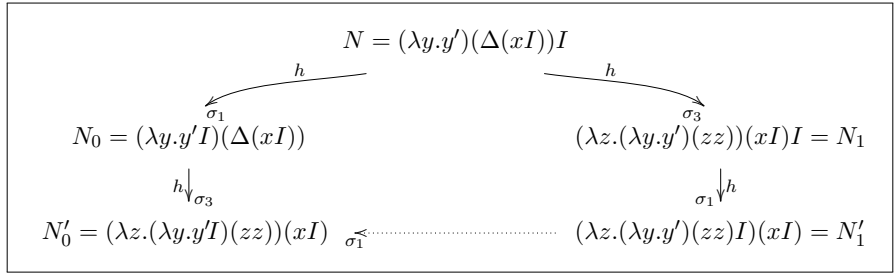

%

Sequentialization (Theorem~\ref{thm:trifurcate}) imposes no order on head $\sigma$-reduction steps, in accor\-dance with the notion of head $\sigma$-reduction (Definition~\ref{def:headsigma}) which puts together head $\sigma_1$/$\sigma_3$-reduction steps.
So, a natural question arises: is it possible to sequentialize them?
More precisely, we wonder if it is possible to anticipate \textit{a priori}  all the head $\sigma_1$- or all the head $\sigma_3$-reduction steps.
The answer is negative, as proved by the next two counterexamples. 

\begin{itemize}
  \item $M = x((\lambda y.z')(zI))\Delta \ToHeadSigmaThree (\lambda y.xz')(zI) \Delta \allowbreak\ToHeadSigmaOne  (\lambda y.xz' \Delta)(zI)= N$,
 but there exists no $L$ such that $M \ToHeadSigmaOneReflTrans L \ToHeadSigmaThreeReflTrans N$.
In fact, $M$ contains only a head $\sigma_3$-redex and $(\lambda y.xz') (zI) \Delta$ 
  has only a head $\sigma_1$-redex, created by firing the head $\sigma_3$-redex in $M$.
  \item $M = x((\lambda y.z')(zI)\Delta) \ToHeadSigmaOne x((\lambda y.z'\Delta)(zI)) \ToHeadSigmaThree (\lambda y.x(z' \Delta))(zI)= N$,
but there is no $L$ such that $M \ToHeadSigmaThreeReflTrans L \ToHeadSigmaOneReflTrans N$. In fact, $M$ contains only a head $\sigma_1$-redex 
and $x((\lambda y.z'\Delta) (zI))$ 
  has only a head $\sigma_3$-redex, created by firing the head $\sigma_1$-redex in $M$. 
\end{itemize}
The impossibility of 
prioritizing a kind of head $\sigma$-reduction over the other is due to the fact that a head $\sigma_1$-reduction step can create a new head $\sigma_3$-redex, and vice-versa.
Thus, sequentialization (and then standardization) does not force a total order on head $\sigma$-redexes.
This 
is not a serious issue, since head $\sigma$-reduction is strongly normalizing (by Proposition~\ref{prop:general-properties}, as $\ToHeadSigma \,\subseteq\, \to_\sigma$) and hence the order in which head $\sigma$-reduction steps are performed is irrelevant.
Moreover, following Remark \ref{rmk:alternative-sigma}, it seems natural to treat head $\sigma_1$- and head $\sigma_3$-reductions as a same reduction also because the two axiom schemes $\sigma_1$ and $\sigma_3$ in the definition of head $\sigma$-reduction (Definition~\ref{def:headsigma}) can be equivalently replaced by the unique axiom scheme
\begin{center}
  {\small
    \AxiomC{}
    \RightLabel{\scriptsize $\sigma$}
    \UnaryInfC{$\ctxL{(\lambda x.M)N} M_1\dots M_m \ToHeadSigma (\lambda x.\ctxL{M})N M_1\dots M_m$}
    \DisplayProof
  }
\end{center}
where $\mathtt{E}$ is a context of the form $\Chole{\cdot}L$ or $V\Chole{\cdot}$ (for any $L \in \Lambda$ and $V \in \Lambda_v$) such that $x \notin \Fv{\mathtt{E}}$.

\medskip
Sequentialization (Theorem~\ref{thm:trifurcate}) says that any $\vg$-reduction sequence from a term $M$ to a term $M'$ can be rearranged into an initial head $\beta_v$-reduction sequence (whose steps 
reduce, in a deterministic way, the unique leftmost-outermost $\beta_v$-redex not under the scope of a $\lambda$) from $M$ to some term  $L$, followed by a head $\sigma$-reduction sequence (whose steps 
reduce, non-deterministically, one of the leftmost-outermost $\sigma$-redexes not in the scope of a $\lambda$) from $L$ to some term $N\!$, followed by an internal $\vg$-reduction sequence from $N$ to $M'\!$.
For this internal $\vg$-reduction sequence, the same kind of decomposition \mbox{can be iterated 
on the subterms of $N$.}

\subsection{Proof of the Sequentialization Theorem}\label{subsect:proof}

In this subsection we present a detailed proof, with all auxiliary lemmas, of Theorem \ref{thm:trifurcate}.
First, we define parallel reduction.

\begin{defi}[Parallel reduction]\label{def:parallel}
  The \emph{parallel reduction} $\Rightarrow \, \subseteq \Lambda \times \Lambda$ is defined inductively by the following rules (where $m \in \Nat$, and $x \notin \Fv{L}$ in the rule $\sigma_1$, $x \notin \Fv{V}$ in the rule~$\sigma_3$):

  \medskip
  \begin{center}
  {\small
    \def\ScoreOverhang{-1pt}
    \def\labelSpacing{1pt}
    \insertBetweenHyps{\hskip 10pt}
    \AxiomC{$V \Rightarrow V'$ \ }
    \AxiomC{$M_i \Rightarrow M_i' \quad (
    \text{for all} \ 0 \leq i \leq m)$}
    \RightLabel{\scriptsize $\beta_v$}
    \BinaryInfC{$(\lambda x.M_0)V\! M_1\dots M_m \!\Rightarrow\! M_0'\sub{V'\!\!}{x} M_1' \dots M_m'$}
    \DisplayProof
    \def\ScoreOverhang{-1pt}
    \def\labelSpacing{1pt}
    \insertBetweenHyps{\hskip 10pt}
    \AxiomC{$N \Rightarrow N'$}
    \AxiomC{$L \Rightarrow L'$}
    \AxiomC{$M_i \Rightarrow M_i' \quad\!\! (
    \text{for all} \ 0 \leq i \leq m)$}
    \RightLabel{\scriptsize $\sigma_1$}
    \TrinaryInfC{$(\lambda x.M_0)N\!L M_1\dots M_m \!\Rightarrow\! (\lambda x.M_0'L')N'\! M_1' \dots M_m'$}
    \DisplayProof

    \def\ScoreOverhang{2pt}
    \def\labelSpacing{2pt}
    \vspace{0.6\baselineskip}
    \AxiomC{$V \Rightarrow V'$}    
    \AxiomC{$N \Rightarrow N'$}
    \AxiomC{$L \Rightarrow L'$}
    \AxiomC{$M_i \Rightarrow M_i' \quad (
    \text{for all} \ 1 \leq i \leq m)$}
    \RightLabel{\scriptsize $\sigma_3$}
    \QuaternaryInfC{$V((\lambda x.L)N) M_1\dots M_m \Rightarrow (\lambda x.V'L')N' M_1' \dots M_m'$}
    \DisplayProof

    \vspace{0.6\baselineskip}
    \AxiomC{$M_i \Rightarrow M_i' \quad (
    \text{for all} \ 0 \leq i \leq m)$}
    \RightLabel{\scriptsize $\lambda$}
    \UnaryInfC{$(\lambda x . M_0) M_1 \dots M_m \Rightarrow (\lambda x . M_0') M_1' \dots M_m'$}
    \DisplayProof
    \qquad\qquad
    \AxiomC{$M_i \Rightarrow M_i' \quad (
    \text{for all} \ 1 \leq i \leq m)$}
    \RightLabel{\scriptsize $\mathit{var}$}
    \UnaryInfC{$x \, M_1 \dots M_m \Rightarrow  x \, M_1' \dots M_m'$}
    \DisplayProof.
  }
  \end{center} 
\end{defi}


The rule $\mathit{var}$, in Definition~\ref{def:parallel},  has no premises when $m = 0$: this is the base case of the inductive definition of $\Rightarrow$.
The rules $\sigma_1$ and $\sigma_3$ have exactly three premises when $m = 0$.
Intuitively, $M \Rightarrow M'$ means that $M'$ is obtained from $M$ by reducing a number of $\beta_v$-, $\sigma_1$- and $\sigma_3$-redexes (existing in $M$) simultaneously.

\begin{defi}[Internal parallel reduction, strong parallel reduction]\label{def:internal}
  The \emph{internal parallel reduction} $\Inter \, \subseteq \Lambda \times \Lambda$ is defined inductively by the following rules ($m \in \Nat$ in the rule $\mathit{right}$):
 \begin{center}
  {\small
    \def\ScoreOverhang{0pt}
    \AxiomC{$N \Rightarrow N'$}
    \RightLabel{\scriptsize $\lambda$}
    \UnaryInfC{$\lambda x.N \Inter \lambda x.N'$}
    \DisplayProof
    \quad
    \AxiomC{}
    \RightLabel{\scriptsize $\mathit{var}$}
    \UnaryInfC{$x \Inter  x$}
    \DisplayProof
    \quad
    \def\ScoreOverhang{0pt}
    \insertBetweenHyps{\hskip 9pt}
    \AxiomC{$V \Rightarrow V'$}
    \AxiomC{$N \Inter N'$}
    \AxiomC{$M_i \Rightarrow M_i' \quad (
	\text{for all} \ 1 \leq i \leq m)$}
    \RightLabel{\scriptsize $\mathit{right}$}
    \TrinaryInfC{$VN M_1 \dots M_m \Inter V'N' M_1' \dots M_m'$}
    \DisplayProof.
  }
  \end{center}\medskip

  \noindent The \emph{strong parallel reduction} $\Rrightarrow \, \subseteq \Lambda \times \Lambda$ is defined by: $M \Rrightarrow N$ iff $M \Rightarrow N$ and there exist $M', M'' \in \Lambda$ such that $M \ToHeadBetavReflTrans M' \ToHeadSigmaReflTrans M'' \Inter N$.
\end{defi}

Notice that the rule $\mathit{right}$ in Definition \ref{def:internal} has exactly two premises when $m = 0$.

\begin{lem}[Reflexivity]\label{rmk:reflexive}
  The relations $\Rightarrow$, $\Rrightarrow$ and $\Inter$ are reflexive. 
\end{lem}
\proof
The reflexivity of $\Rrightarrow$ follows immediately from the reflexivity of $\Rightarrow$ and $\Inter$. 
  The proofs of reflexivity of $\Rightarrow$ and $\Inter$ are both by structural induction on a term: 
  in the case of $\Rightarrow$, recall that any term is of the form $(\lambda x.N) M_1 \dots M_m$ or $x \, M_1 \dots M_m$ for some $m \in \Nat$ (Remark~\ref{rem:shape}), and then apply the rule $\lambda$ or $\mathit{var}$ respectively, together with the inductive hypothesis; 
  in the case of $\Inter$, recall that every term is of the form $\lambda x.M$ or $x$ or $VN M_1 \dots M_m$ for some $m \in \Nat$, and then apply the rule $\lambda$ (together with the reflexivity of $\Rightarrow$) or $\mathit{var}$ or $\mathit{right}$ (together with the reflexivity of $\Rightarrow$ and the inductive hypothesis) respectively.
\qed

We have $\Inter \, \subsetneq \, \Rrightarrow \, \subseteq \, \Rightarrow$ (first, prove that $\Inter \, \subseteq \, \Rightarrow$ by induction on the derivation of $N \Inter N'\!$, the other inclusions follow from the definition of $\Rrightarrow$; note that $II \Rrightarrow I$ but $II \not\Inter I$) and, by reflexivity of $\Rightarrow$ (Lemma~\ref{rmk:reflexive}), $\ToHeadBetav \, \subsetneq \, \Rightarrow$ and $\ToHeadSigma \, \subsetneq \, \Rightarrow$.
Observe that $\Delta\Delta \ \mathsf{R} \ \Delta\Delta$ for any $\mathsf{R} \in \{\mapsto_{\beta_v}, \ToHeadBetav, \to_{\beta_v}, \Rightarrow, \Inter, \Rrightarrow\}$, even if for different reasons: for example, $\Delta \Delta \Inter \Delta \Delta$ by reflexivity of $\Inter$ (Lemma~\ref{rmk:reflexive}), whereas $\Delta\Delta \ToHeadBetav \Delta\Delta$ by reducing the 
only $\beta_v$-redex.


Some useful properties relating values and reductions follow.
Note that Lemmas~\ref{rmk:abs}.\ref{rmk:abs.novalue}-\ref{rmk:abs.notovalue} imply that \emph{all values are head $\vg$-normal}; the converse 
fails, as $xI$ is head $\vg$-normal but not a value.
\cite{guerrieri15wpte} proves that all \emph{closed} head $\vg$-normal forms are values (in fact, abstractions).

\begin{lem}[Values vs.\ reductions]\label{rmk:abs}\hfill
  \begin{enumerate}
    \item\label{rmk:abs.novalue} The head $\beta_v$-reduction $\ToHeadBetav$ does not reduce a value (i.e.~values are head $\beta_v$-normal).
    \item\label{rmk:abs.notovalue} The head $\sigma$-reduction $\ToHeadSigma$ does neither reduce a value nor reduce to a value.

    \item\label{rmk:abs.anti} Variables and abstractions are preserved 
    by $\overset{\scriptscriptstyle\mathit{int}}{\Leftarrow}$
    , more precisely: if $M \Inter x$ (resp.\ $M \Inter \lambda x.N'$) then $M = x$ (resp.~$M = \lambda x.N$ for some $N \in \Lambda$ such that $N \Rightarrow N'$).

    \item\label{rmk:abs.parallel} If $M \Rightarrow M'$ then $\lambda x.M \ \mathsf{R} \ \lambda x.M'$ for any $\mathsf{R} \in \{\Rightarrow, \Inter, \Rrightarrow\}$. 

    \item\label{rmk:abs.value} For any $V, V' \in \Lambda_v$, one has $V \Inter V'$ iff $V \Rightarrow V'$ iff $V \Rrightarrow V'$. 
  \end{enumerate}
\end{lem}
\proof\hfill
  \begin{enumerate}
  \item For every $M \in \Lambda$ and every $V \in \Lambda_v$, we have $V \not\ToHeadBetav M$, 
 because the head $\beta_v$-reduction does not reduce under $\lambda$'s.
  \item  For any $N \in \Lambda$ and $V \in \Lambda_v$, $V \not\ToHeadSigma N$ and $N \not\ToHeadSigma V$ since
  in the conclusion of any rule of Definition~\ref{def:headsigma} the terms on the right and on the left of $\ToHeadSigma$ are applications.
\item By simple inspection of the 
rules of $\,\Inter$ (Definition~\ref{def:internal}), if $M \Inter x$ (resp.~$M \Inter \lambda x.N'$) then  the last rule in the derivation is necessarily \textit{var} (resp.~$\lambda$).
\item For $\mathsf{R} \in \{\Rightarrow, \Inter\}$ we apply the rule $\lambda$ to conclude that $\lambda x.M \ \mathsf{R} \ \lambda x.M'$, therefore $\lambda x.M \Rrightarrow \lambda x.M'$ according to the definition of $\Rrightarrow$, since $\ToHeadBetavReflTrans$ and $\ToHeadSigmaReflTrans$ are reflexive.
\item First we show that $V \!\Inter V'$ iff $V \!\Rightarrow V'\!$.
  The left-to-right direction \mbox{holds since $\Inter \, \subseteq \, \Rightarrow$.}
  Conversely, assume $V \!\Rightarrow\! V'$: if $V$ is a variable then $V = V'$ and hence $V \Inter V'$ by applying the rule $\mathit{var}$ for $\Inter$; otherwise $V = \lambda x.N$ for some $N \in \Lambda$, and then necessarily $V' = \lambda x.N'$ with $N \Rightarrow N'$, so $V \Inter V'$ by applying the rule $\lambda$ for $\Inter$.
  
  Now we prove that $V \Rightarrow V'$ iff $V \Rrightarrow V'$. The right-to-left direction follows immediately from the definition of $\Rrightarrow$ (Definition~\ref{def:internal}). 
  Conversely, if $V \Rightarrow V'$ then we have just shown that $V \Inter V'$, so $V \Rrightarrow V'$ since $\ToHeadBetavReflTrans$ and $\ToHeadSigmaReflTrans$ are reflexive.
\qed
  \end{enumerate}



\noindent We collect some basic closure properties and relations that hold for reductions.

\begin{lem}[Properties of parallel, head and internal reductions]\hfill
\label{rmk:preliminary}
  \begin{enumerate}
    \item\label{rmk:preliminary.appparallel} If $M \Rightarrow M'$ and $N \Rightarrow N'$ then $MN \Rightarrow M'N'$.

    \item\label{rmk:preliminary.apphead} If $\mathsf{R} \in \{\ToHeadBetav, \ToHeadSigma \}$ 
                 and $M \ \mathsf{R} \ M'$, then $MN \ \mathsf{R} \ M'N$ for any $N \in \Lambda$.

    \item\label{rmk:preliminary.appinternal} If $M \Inter M'$ and $N \Rightarrow N'$ where     $M' \notin \Lambda_v$, then $MN \Inter \ M'N'$.

    \item\label{rmk:preliminary.preConfluence}
    $\to_\vg \, \subseteq \, \Rightarrow \, \subseteq \to_\vg^*$ and hence  $\Rightarrow^* \, = \, \to_{\mathsf{v}}^*$. 
    
    \item\label{rmk:preliminary.interinclusion} $\ToInt \, \subseteq \, \Inter \, \subseteq \, \ToIntTrans$ and hence $\InterTrans =\, \ToIntTrans$.
    
    \item\label{rmk:preliminary.confluence} $\Rightarrow$ is confluent. 
    
    \item\label{rmk:preliminary.interExpansion} If $M \ToIntTrans x$ (resp.\ $M \ToIntTrans \lambda x.N'$) then $M = x$ (resp.\ $M = \lambda x.N$ with $N \to_\vg^* N'$).

    \item\label{rmk:preliminary.headsubstitution} For any $\mathsf{R} \in \{\ToHeadBetav, \ToHeadSigma\}$, if $M \ \mathsf{R} \ M'$ then $M\sub{V}{x} \ \mathsf{R} \ M'\sub{V}{x}$ for any $V \in \Lambda_v$. 
  \end{enumerate}
\end{lem}
\proof\hfill
\begin{enumerate}
\item 
  Just add the derivation of $N \Rightarrow N'$ as the ``rightmost'' premise of the last rule of the derivation of $M \Rightarrow M'$.
  \item 
  In the 
  conclusion of the derivation of $M \ \mathsf{R} \ M'$, replace $M \ \mathsf{R} \ M'$ with $MN \ \mathsf{R} \ M'N$.

\item 
The last rule in the derivation of $M \Inter M'$ can be neither $\lambda$ nor $\mathit{var}$ because $M' \notin \Lambda_v$,
so it is $\mathit{right}$ and hence we can add 
the derivation of $N \Rightarrow N$ (which exists since $\Rightarrow$ is reflexive, Lemma~\ref{rmk:reflexive}) as its rightmost premise. 
Note that the hypothesis $M' \notin \Lambda_v$ is crucial: for example, $x \Inter x$ and $I\Delta  \Rightarrow \Delta$ but $I\Delta  \not\Inter \Delta$ and thus $x(I\Delta)\not \Inter x\Delta$.
\item 
The proof that $M \to_\vg M'$ implies $M \Rightarrow M'$ is by induction on $M \in \Lambda$, using the reflexivity of $\Rightarrow$ (Lemma~\ref{rmk:reflexive}) and Lemma~\ref{rmk:preliminary}.\ref{rmk:preliminary.appparallel}. 
The proof that $M \Rightarrow M'$ implies $M \to_\vg^* M'$ is by straightforward induction on the derivation of $M \Rightarrow M'$. 
\item 
We prove that $M \ToInt M'$ implies $M \Inter M'$ by induction on $M \in \Lambda$. 
    According to Remark~\ref{rem:shape}, $M = VN_1 \dots N_n$ for some $n \in \Nat$, $V \in \Lambda_v$ and $N_1, \dots, N_n \in \Lambda$. 
    Since $M \to_\vg M'$ and $M \not\ToHead M'$, there are only three cases:
    \begin{itemize}
      \item either $M' = V'N_1 \dots N_n$ with $V \to_\vg V'$, then 
      $V = \lambda x.N$ and $V' = \lambda x.N'$ with $N \to_\vg N'$ by Remark~\ref{rmk:fromvalue}, so $N \Rightarrow N'$ according to Lemma~\ref{rmk:preliminary}.\ref{rmk:preliminary.preConfluence}, and thus $V = \lambda x.N \Inter \lambda x.N' = V'$ by applying the rule $\lambda$ for $\Inter$; if $n = 0$ then $M = \lambda x.N$ and $M' = \lambda x.N'$ and we are done; otherwise $n > 0$ and hence $V \Rightarrow V'$ (Lemma~\ref{rmk:abs}.\ref{rmk:abs.value}), so $M \Inter M'$ by applying the rule $\mathit{right}$ for $\Inter$, since $N_1 \Inter N_1$ and $N_i \Rightarrow N_i$ for any $2 \leq i \leq n$ by reflexivity of $\Inter$ and $\Rightarrow$ (Lemma~\ref{rmk:reflexive});
      \item or $n > 0$ and $M' = VN_1'N_2 \dots N_n$ with $N_1 \ToInt N_1'$, then $N_1 \Inter N_1'$ by induction hypothesis, and $V \Rightarrow V$ and $N_i \Rightarrow N_i$ for any $2 \leq i \leq n$ by reflexivity of $\Rightarrow$ (Lemma~\ref{rmk:reflexive}); hence $M \Inter M'$ by applying the rule $\mathit{right}$ for $\Inter$;   
      \item or $n > 1$ and $M' = VN_1 \dots N_i' \dots N_n$ with $N_i \to_\vg N_i'$ for some $2 \leq i \leq n$, then $N_i \Rightarrow N_i'$ by Lemma~\ref{rmk:preliminary}.\ref{rmk:preliminary.preConfluence}, $V \Rightarrow V$ and $N_j \Rightarrow N_j$ for any $2 \leq j \leq n$ with $j \neq i$ and $N_1 \Inter N_1$ by reflexivity of $\Rightarrow$ and $\Inter$ (Lemma~\ref{rmk:reflexive}); hence $M \Inter M'$ by applying the rule $\mathit{right}$ for $\Inter$.
    \end{itemize}

    \noindent The proof that $M \Inter M'$ implies $M \ToIntTrans M'$ is by straightforward induction on the derivation of $M \Inter M'$, using that $\Rightarrow \, \subseteq \, \to_\vg^*$ (Lemma~\ref{rmk:preliminary}.\ref{rmk:preliminary.preConfluence}). 
\item Since $\Rightarrow^* \, = \, \to_{\mathsf{v}}^*$ according to  Lemma~\ref{rmk:preliminary}.\ref{rmk:preliminary.preConfluence}, Proposition~\ref{prop:general-properties} just says that $\Rightarrow$ is confluent.
Anyway, we remark that $\Rightarrow$ does not enjoy the diamond property, see Section~\ref{sect:conclusions}.
\item Since $\InterTrans \, = \ToIntTrans$ (Lemma~\ref{rmk:preliminary}.\ref{rmk:preliminary.interinclusion}) and $\Rightarrow \,\subseteq \,\to_\vg^*$ (Lemma~\ref{rmk:preliminary}.\ref{rmk:preliminary.preConfluence}), then Lemma~\ref{rmk:abs}.\ref{rmk:abs.anti} can be reformulated substituting $\ToIntTrans$ for $\Inter$,  and $\to_\vg^*$ for $\Rightarrow$.
\item The proof is by induction on the derivation of $M \ \mathsf{R} \ M'$, for any $\mathsf{R} \in \{\ToHeadBetav, \ToHeadSigma \}$.
\qed 
\end{enumerate}

\noindent Parallel reduction is closed under substitution, as stated by the following lemma.

\begin{lem}[Substitution vs. $\Rightarrow$]\label{lemma:parallelsubstitution}
  If $M \!\Rightarrow\! M'$ and $V \!\Rightarrow\! V'$ 
  then $M\sub{V}{x} \!\Rightarrow\! M'\sub{V'\!}{x}$.
\end{lem}
\proof
   By induction on the derivation of $M \Rightarrow M'$.
    Let us consider its last rule $\mathsf{r}$.

  \begin{itemize}
  \item If $\mathsf{r} = \mathit{var}$ then $M = y \, M_1 \dots M_m$ and
    $M' = y \, M_1' \dots M_m'$ with $m \in \Nat$ and
    $M_i \Rightarrow M_i'$ for any $1 \leq i \leq m$. By induction
    hypothesis, $M_i\sub{V\!}{x} \Rightarrow M_i'\sub{V'\!}{x}$ for any
    $1 \leq i \leq m$. If $y \neq x$ then
    $M \sub{V}{x} = y \, M_1 \sub{V\!}{x} \dots M_m\sub{V\!}{x}$ and
    $M' \sub{V'\!}{x} = y \, M_1' \sub{V'\!}{x} \dots M_m'\sub{V'\!}{x}$,
    so $M\sub{V\!}{x} \Rightarrow M'\sub{V'\!}{x}$ by applying the rule
    $\mathit{var}$ for $\Rightarrow$. Otherwise $y = x$ and then
    $M \sub{V\!}{x} = V M_1 \sub{V\!}{x} \dots M_m\sub{V\!}{x}$ and
    $M' \sub{V'\!}{x} = V' M_1' \sub{V'\!}{x} \dots M_m'\sub{V'\!}{x}$,
    hence $M\sub{V\!}{x} \Rightarrow M'\sub{V'\!}{x}$ by
    Lemma~\ref{rmk:preliminary}.\ref{rmk:preliminary.appparallel}.

   \item If $\mathsf{r} = \lambda$ then $M = (\lambda y.M_0) M_1 \dots M_m$
    and $M' = (\lambda y. M_0') M_1' \dots M_m'$ with $m \in \Nat$ and
    $M_i \Rightarrow M_i'$ for all $0 \leq i \leq m$; we can suppose
    without loss of generality that $y \notin \Fv{V} \cup \{x\}$. By
    induction hypothesis, $M_i\sub{V\!}{x} \Rightarrow M_i' \sub{V'\!}{x}$
    for all $0 \leq i \leq m$. 
    By applying the rule $\lambda$ for $\Rightarrow$, 
    $M \sub{V}{x} = (\lambda y.M_0\sub{V}{x}) M_1 \sub{V}{x} \dots
    M_m\sub{V}{x} \Rightarrow (\lambda y.M_0'\sub{V'\!}{x}) M_1'
    \sub{V'\!}{x} \dots M_m'\sub{V'\!}{x} = M' \sub{V'\!}{x}$.

 \item   
    If $\mathsf{r} = \sigma_1$ then
    $M = (\lambda y.M_0)NL M_1 \dots M_m$ and
    $M' = (\lambda y. M_0'L')N' M_1' \dots M_m'$
    with 
    $m \in \Nat$, $L \Rightarrow L'$, $N \Rightarrow N'$ and
    $M_i \Rightarrow M_i'$ for any $0 \leq i \leq m$; we can suppose
    without loss of generality that $y \notin \Fv{V} \cup \{x\}$.  By
    induction hypothesis, $L\sub{V}{x} \Rightarrow L' \sub{V'}{x}$,
    $N\sub{V}{x} \Rightarrow N' \sub{V'}{x}$ and
    $M_i\sub{V}{x} \Rightarrow M_i' \sub{V'}{x}$ for any
    $0 \leq i \leq m$.  Hence
    $M\sub{V\!}{x} \Rightarrow M'\sub{V'\!}{x}$ by applying 
    rule~$\sigma_1$, since
    $M \sub{V\!}{x} = (\lambda y.M_0\sub{V\!}{x}) N\sub{V\!}{x}
    L\sub{V\!}{x} M_1\sub{V\!}{x} \dots M_m\sub{V\!}{x}$
    and
    $M'\sub{V'\!}{x} = (\lambda
    y.M_0'\sub{V'\!}{x}L'\sub{V'\!}{x})N'\sub{V'\!}{x}
    M_1'\sub{V'\!}{x} \dots M_m'\sub{V'\!}{x}$.

   \item   If $\mathsf{r} = \sigma_3$ then
    $M = U ((\lambda y.L)N) M_1 \dots M_m$ and
    $M' = (\lambda y. U'L')N' M_1' \dots M_m'$
    with 
    $m \in \Nat$ and $U,U' \in \Lambda_v$, $U \Rightarrow U'$,
    $L \Rightarrow L'$, $N \Rightarrow N'$ and $M_i \Rightarrow M_i'$
    for any $1 \leq i \leq m$; without loss of
    generality we can assume
    $y \notin\allowbreak \Fv{V} \allowbreak \cup \{x\}$.  By induction
    hypothesis, $U\sub{V}{x} \Rightarrow U' \sub{V'\!}{x}$
    , $L\sub{V}{x} \Rightarrow L' \sub{V'\!}{x}$,
    $N\sub{V}{x} \Rightarrow N' \sub{V'\!}{x}$ and
    $M_i\sub{V}{x} \Rightarrow M_i' \sub{V'\!}{x}$ for any
    $1 \leq i \leq m$.  
    So, $M\sub{V\!}{x} \Rightarrow M'\sub{V'\!}{x}$ by applying the
    rule~$\sigma_3$, since $M\sub{V}{x} = \allowbreak U\sub{V}{x} ((\lambda y.L\sub{V}{x})
    N\sub{V}{x}) M_1 \sub{V}{x} \dots M_m\sub{V}{x}$ and 
    $M' \sub{V'\!}{x} = (\lambda y.U'\sub{V'\!}{x} L'\sub{V'\!}{x})
    N'\sub{V'\!}{x} M_1' \sub{V'\!}{x} \dots M_m'\sub{V'\!}{x}$ with $U\sub{V}{x}, U' \sub{V'\!}{x} \in \Lambda_v$.

  \item    Finally, if $\mathsf{r} = \beta_v$, then
    $M = (\lambda y. M_0)V_0 M_1\dots M_m$ and
    $M' = M_0'\sub{V_0'}{y} M_1' \dots M_m'$ with $m \in \Nat$,
    $V_0 \Rightarrow V_0'$ and $M_i \Rightarrow M_i'$ for any
    $0 \leq i \leq m$; we can suppose without loss of generality that
    $y \notin \Fv{V} \cup \{x\}$.  By induction hypothesis,
    $V_0\sub{V\!}{x} \Rightarrow V_0'\sub{V'\!}{x}$ and
    $M_i\sub{V\!}{x} \Rightarrow M_i' \sub{V'\!}{x}$ for any
    $0 \leq i \leq m$.  So,
    $M\sub{V\!}{x} \Rightarrow M'\sub{V'\!}{x}$ by apply\-ing the rule
    $\beta_v$, since
    $M \sub{V}{x} = (\lambda y.M_0\sub{V\!}{x}) V_0\sub{V\!}{x}
    M_1 \sub{V\!}{x} \dots M_m\sub{V\!}{x}$
    and
    \begin{align*}
      M' \sub{V'\!}{x} &= M_0'\sub{V_0'}{y}\sub{V'\!}{x} M_1'
      \sub{V'\!}{x} \dots M_m'\sub{V'\!}{x} \\
      &= M_0'\sub{V'\!}{x}\sub{V_0'\sub{V'\!}{x}}{y} M_1'
      \sub{V'\!}{x} \dots M_m'\sub{V'\!}{x}.
  \end{align*}
  
  \vspace{-\baselineskip}
  \qed
\end{itemize}

\noindent The following lemma will play a crucial role in the proof of Lemmas~\ref{lemma:strongparallelsubstitution}-\ref{lemma:main} and 
shows that head $\sigma$-reduction $\ToHeadSigma$ can be postponed to head $\beta_v$-reduction $\ToHeadBetav$.

\begin{lem}[Commutation of head reductions]\label{lemma:commutation}\hfill
  \begin{enumerate}
    \item\label{lemma:commutation.onestep} If $M \ToHeadSigma L \ToHeadBetav N$ then there exists $L' \in \Lambda$ such that $M \ToHeadBetav L' \ToHeadSigmaRefl N$. 
    \item\label{lemma:commutation.transitive} If $M \ToHeadSigmaReflTrans L \ToHeadBetavReflTrans N$ then there exists $L' \in \Lambda$ such that $M \ToHeadBetavReflTrans L' \ToHeadSigmaReflTrans N$. 
    \item\label{lemma:commutation.general} If $M \ToHeadReflTrans M'$ then there exists $N \in \Lambda$ such that $M \ToHeadBetavReflTrans N \ToHeadSigmaReflTrans M'$.
  \end{enumerate}
\end{lem}
\proof\hfill
  \begin{enumerate}
    \item   By induction on the derivation of $M \ToHeadSigma L$. Let us consider its last rule $\mathsf{r}$.
      \begin{itemize}
      \item If $\mathsf{r} = \sigma_1$ then $M = (\lambda x.M_0)N_0L_0
        M_1 \dots M_m$ and $L = (\lambda x.M_0L_0)N_0 M_1 \dots M_m$
        where $m \in \Nat$ and $x \notin \Fv{L_0}$.  Since $L
        \ToHeadBetav N$, there are only two cases:
        \begin{itemize}
        \item either $N_0 \ToHeadBetav N_0'$ and $N = (\lambda
          x.M_0L_0)N_0' M_1 \dots M_m$ (according to the rule
          $\mathit{right}$ for $\ToHeadBetav$), then $M \ToHeadBetav
          (\lambda x.M_0)N_0'L_0 M_1 \dots M_m \ToHeadSigma N$;
        \item or $N_0 \in \Lambda_v$ and $N = M_0\sub{N_0}{x} L_0 M_1 \!\dots M_m$ 
        (according to the rule ${\beta_v}$ for $\ToHeadBetav$, since $ x \notin \Fv{L_0}$), therefore $M \ToHeadBetav N$.
        \end{itemize}
        \item If $\mathsf{r} = \sigma_3$ then $M = V((\lambda x.L_0)N_0) M_1
        \dots M_m$ and $L = (\lambda x.VL_0)N_0 M_1 \dots M_m$ with
        \mbox{$m \in \Nat$} and $x \notin \Fv{V}$.  Since $L
        \ToHeadBetav N$, there are only two cases:
        \begin{itemize}
        \item either $N_0 \ToHeadBetav N_0'$ and $N = (\lambda
          x.VL_0)N_0' M_1 \dots M_m$ (according to the rule
          $\mathit{right}$ for $\ToHeadBetav$), then $M \ToHeadBetav
          V((\lambda x.L_0)N_0') M_1 \dots M_m \allowbreak\ToHeadSigma
          N$;
        \item or $N_0 \in \Lambda_v$ and $N = V L_0\sub{N_0}{x} M_1 \dots M_m$ (according to the rule ${\beta_v}$ for $\ToHeadBetav$,
          because $x \notin \Fv{V}$), so $M \ToHeadBetav N$.
        \end{itemize}
        \item Finally, if $\mathsf{r} = \mathit{right}$ then $M = VN_0 M_1
        \dots M_m$ and $L = VN_0' M_1 \dots M_m$ with $m \in \Nat$ and
        $N_0 \ToHeadSigma N_0'$.  By
        Lemma~\ref{rmk:abs}.\ref{rmk:abs.notovalue}, $N_0' \notin
        \Lambda_v$ and thus, since $L \ToHeadBetav N$, the only
        possibility is that $N_0' \ToHeadBetav N_0''$ and $N = VN_0''
        M_1 \dots M_m$ (according to the rule $\mathit{right}$ for
        $\ToHeadBetav$).  By induction hypothesis, there exists
        $N_0''' \in \Lambda$ such that $N_0 \ToHeadBetav N_0'''
        \ToHeadSigmaRefl N_0''$. Therefore, $M \ToHeadBetav VN_0'''
        M_1 \dots M_m \ToHeadSigmaRefl N$.
      \end{itemize}

   \item 
   By hypothesis, there exist $m,n \in \Nat$ and $M_0, \dots, M_m, N_0, \dots, N_n \in \Lambda$ such that 
   $M = M_0 \ToHeadSigma \dots \ToHeadSigma M_m = L = N_0 \ToHeadBetav \dots \ToHeadBetav N_n = N$. 
    We prove by induction on $m \in \Nat$ that 
    $M \ToHeadBetavReflTrans L' \ToHeadSigmaReflTrans N$ for some $L' \in \Lambda$.
    \begin{itemize}
    \item If $m = 0$ (resp.~$n = 0$) then we conclude by taking
      $L' = N$ (resp.~$L' = M$).

    \item  Suppose $m,n > 0$: by applying
      Lemma~\ref{lemma:commutation}.\ref{lemma:commutation.onestep} at
      most $n$ times, there exist
      $N_0', \dots, N_{n-1}' \allowbreak\in \Lambda$ such that
      $M_{m-1} \ToHeadBetav N_0' \ToHeadBetavRefl \dots \ToHeadBetavRefl N_{n-1}' \ToHeadRefl N$.
      By induction hypothesis (applied to $M = M_0 \ToHeadSigma \dots \ToHeadSigma M_{m-1} \ToHeadBetavReflTrans N$ or $M = M_0 \ToHeadSigma \dots \ToHeadSigma M_{m-1} \ToHeadBetavReflTrans N_{n-1}'$ depending on whether $N_{n-1}' \ToHeadBetav N$ or $N_{n-1}' \ToHeadSigmaRefl N$, respectively), there exists $L' \in \Lambda$ such that $M \ToHeadBetavReflTrans L' \ToHeadSigmaReflTrans N$.
    \end{itemize}

    \item 
By hypothesis, there exist $n \in \Nat$ and $L, M_1, N_1, \dots, M_n, N_n \in \Lambda$ such that 
$M \ToHeadBetavReflTrans L \ToHeadSigmaTrans M_1 \ToHeadBetavTrans N_1 \ToHeadSigmaTrans \dots \ToHeadBetavTrans N_{n-1} \ToHeadSigmaTrans M_{n} \ToHeadBetavTrans N_n \ToHeadSigmaReflTrans M'$ 
(i.e. $n$ is the number of subsequences of the shape $\ToHeadSigma\ToHeadBetav$ in the head $\vg$-reduction sequence from $M$ to $M'$). 
    We prove by induction on $n \in \Nat$ that $M \ToHeadBetavReflTrans N \ToHeadSigmaReflTrans M'$ for some $N \in \Lambda$.
    \begin{itemize}
    \item If $n = 0$ then $M \ToHeadBetavReflTrans L \ToHeadSigmaReflTrans M'$ and hence we conclude by
      taking $N = L$.
     \item Suppose $n > 0$. By applying the induction hypothesis to the head $\vg$-reduction sequence from $M$ to $M_n$,
      $M \ToHeadBetavReflTrans N' \ToHeadSigmaReflTrans M_n \ToHeadBetavTrans N_n \ToHeadSigmaReflTrans M'$ for some $N' \in \Lambda$.
      By Lemma~\ref{lemma:commutation}.\ref{lemma:commutation.transitive},
      $M \ToHeadBetavReflTrans N' \ToHeadBetavReflTrans N \ToHeadSigmaReflTrans N_n \ToHeadSigmaReflTrans M'$ for some $N \in \Lambda$.
\qed
    \end{itemize}
  \end{enumerate}


\noindent We are now 
ready to 
retrace Takahashi's method \cite{Takahashi95} in our setting with $\beta_v$- and $\sigma$-reductions.
The next four lemmas govern strong parallel reduction and will be used to prove Lemma\,\ref{lemma:main}, the key lemma stating that $\Rightarrow$ can be ``sequentialized'' according to $\Rrightarrow$. 

\begin{lem}\label{lemma:nonvalue}
  If $M \Rrightarrow M'$ and $N \Rightarrow N'$ and $M' \notin \Lambda_v$, then $MN \Rrightarrow M'N'$.
\end{lem}
\begin{proof}
  From the definition of $M \Rrightarrow M'$ it follows that $M \Rightarrow M'$ and $M \ToHeadBetavReflTrans L \ToHeadSigmaReflTrans L' \Inter M'$ for some $L, L' \in \Lambda$.
  Hence, $MN \Rightarrow M'N'$ by Lemma~\ref{rmk:preliminary}.\ref{rmk:preliminary.appparallel}, and 
  $MN \ToHeadBetavReflTrans LN \ToHeadSigmaReflTrans L'N$ by Lemma~\ref{rmk:preliminary}.\ref{rmk:preliminary.apphead}.
  Since $M' \notin \Lambda_v$, $L'N \Inter M'N'$ by Lemma~\ref{rmk:preliminary}.\ref{rmk:preliminary.appinternal}. Therefore, $MN \Rrightarrow M'N'$.
\end{proof}

\begin{lem}[Applicative closure of $\Rrightarrow$]\label{lemma:appstrong}
  If $M \Rrightarrow M'$ and $N \Rrightarrow N'$ then $MN \Rrightarrow M'N'$.
\end{lem}
\begin{proof}
  If $M' \notin \Lambda_v$ 
  then $MN \Rrightarrow M'N'$ by Lemma~\ref{lemma:nonvalue}, since $N \Rrightarrow N'$ implies $N \Rightarrow N'$.
  
  Assume $M' \!\in \Lambda_v$: $MN \Rightarrow M'N'$ by Lemma~\ref{rmk:preliminary}.\ref{rmk:preliminary.appparallel}, since $M \Rightarrow M'$ and $N \Rightarrow N'\!$.
  By hypothesis, there are  $M_0,  M_0', N_0, N_0' \in \Lambda$ such that
  $M \ToHeadBetavReflTrans M_0 \ToHeadSigmaReflTrans M_0' \Inter M'$ and $N \ToHeadBetavReflTrans N_0 \ToHeadSigmaReflTrans N_0' \Inter N' \,$. 
  By Lemma~\ref{rmk:abs}.\ref{rmk:abs.anti}, 
  $M_0' \in \Lambda_v$ since $M' \in \Lambda_v$, thus $M_0 = M_0'$ by Lemma~\ref{rmk:abs}.\ref{rmk:abs.notovalue} (and $M_0 \Rightarrow M'$ since $\Inter \, \subseteq \, \Rightarrow$). Since $M_0 \in \Lambda_v$, using the rules $\mathit{right}$ for $\ToHeadBetav$ and $\ToHeadSigma$, we have $M_0N \ToHeadBetavReflTrans  M_0N_0$ and $M_0N_0 \ToHeadSigmaReflTrans  M_0N_0'$. 
  By Lemma~\ref{rmk:preliminary}.\ref{rmk:preliminary.apphead}, $MN \ToHeadBetavReflTrans M_0N$.
  By applying the rule $\mathit{right}$ for $\Inter$, we have $M_0N_0' \Inter M'N'$\!. 
  Therefore, $MN \ToHeadBetavReflTrans M_0N \ToHeadBetavReflTrans M_0 N_0 \ToHeadSigmaReflTrans M_0 N_{0}' \Inter M'N'$ and hence $MN \Rrightarrow M'N'$.
\end{proof}

\begin{lem}[Substitution vs. $\Inter$]
\label{lemma:bastard}
  If $M \Inter M'$ and $V \Inter V'$ then $M\sub{V}{x} \Inter M'\sub{V'\!}{x}$.
\end{lem}
\begin{proof}
  By induction on $M \in \Lambda$. Let us consider the last rule $\mathsf{r}$ of the derivation of $M \Inter M'$.
\begin{itemize}
\item  If $\mathsf{r} = \mathit{var}$ then $M = M'$ and there are only two cases: either $M = x$ and then $M\sub{V\!}{x} = V \Inter V' = M'\sub{V'\!}{x}$; or $M = y \neq x$  and then $M\sub{V\!}{x} = y = M'\sub{V'\!}{x}$, therefore $M\sub{V\!}{x} \Inter M'\sub{V'\!}{x}$ by reflexivity of $\Inter$ (Lemma~\ref{rmk:reflexive}).

\item  If $\mathsf{r} = \lambda$ then $M = \lambda y.N$ and $M' = \lambda y.N'$ with $N \Rightarrow N'$; we can suppose without loss of generality that $y \notin \Fv{V} \cup \{x\}$. 
  We have $N\sub{V}{x} \Rightarrow N'\sub{V'\!}{x}$ according to Lemma~\ref{lemma:parallelsubstitution}, since $V \Inter V'$ implies $V \Rightarrow V'$ (Lemma~\ref{rmk:abs}.\ref{rmk:abs.value}). 
  By applying the rule $\lambda$ for $\Inter$, we have $M\sub{V}{x} = \lambda y.N\sub{V}{x} \Inter \lambda y.N'\sub{V'\!}{x} = M'\sub{V'\!}{x}$
  .

\item  Finally, if $\mathsf{r} = \mathit{right}$ then $M = UN M_1\dots M_m$ and $M' = U'N' M_1'\dots M_m'$ for some $m \in \Nat$ with $U, U' \in \Lambda_v$ such that $U \Rightarrow U'$, $N \Inter N'$ and $M_i \Rightarrow M_i'$ for any $1 \leq i \leq m$. 
  By induction hypothesis, $N\sub{V}{x} \Inter N'\sub{V'}{x}$.
  By Lemma~\ref{lemma:parallelsubstitution}, $U\sub{V}{x} \Rightarrow U'\sub{V'\!}{x}$ and $M_i\sub{V}{x} \Rightarrow M_i'\sub{V'\!}{x}$ for any $1 \leq i \leq m$, since $V \Inter V'$ implies $V \Rightarrow V'$ (Lemma~\ref{rmk:abs}.\ref{rmk:abs.value}).
  By applying the rule $\mathit{right}$ for $\Inter$ (note that $U\sub{V}{x}, U' \sub{V'\!}{x} \in \Lambda_v$), we have
 \begin{align*}
M\sub{V}{x} &=  U\sub{V}{x} N\sub{V}{x} M_1\sub{V}{x} \dots M_m\sub{V}{x} \\
& \Inter U'\sub{V'\!}{x} N'\sub{V'\!}{x} M_1'\sub{V'\!}{x} \dots M_m'\sub{V'\!}{x} = M'\sub{V'\!}{x}.
  \end{align*}

\vspace{-\baselineskip}
\qedhere
\end{itemize}
\end{proof}

\noindent Lemma~\ref{lemma:bastard} is used to prove the following substitution lemma for $\Rrightarrow$
.

\begin{lem}[Substitution vs. $\Rrightarrow$]\label{lemma:strongparallelsubstitution} 
 If $M \!\Rrightarrow\! M'$ and $V \!\Rrightarrow\! V'\!$ then $M\sub{V\!}{x} \Rrightarrow M'\sub{V'\!}{x}$.
\end{lem}
\begin{proof}
  According to Lemma~\ref{lemma:parallelsubstitution}, $M\sub{V}{x} \Rightarrow M'\sub{V'\!}{x}$ since $M \Rightarrow M'$ and $V \Rightarrow V'$.
  By hypothesis, there exist $N, L \in \Lambda$ such that $M \ToHeadBetavReflTrans L \ToHeadSigmaReflTrans N \Inter M'$.
  By Lemma~\ref{rmk:preliminary}.\ref{rmk:preliminary.headsubstitution}, $M\sub{V}{x} \ToHeadBetavReflTrans L\sub{V}{x}$ and $L\sub{V}{x} \ToHeadSigmaReflTrans  N\sub{V}{x}$.
  By Lemma~\ref{lemma:bastard} (since $V \Rrightarrow V'$ implies $V \Inter V'$ according to Lemma~\ref{rmk:abs}.\ref{rmk:abs.value}), we have $N\sub{V}{x} \Inter M'\sub{V'\!}{x}$, thus $M\sub{V}{x} \ToHeadBetavReflTrans L\sub{V}{x} \ToHeadSigmaReflTrans N\sub{V}{x} \Inter M'\sub{V'\!}{x}$ and 
  therefore $M\sub{V}{x} \Rrightarrow M'\sub{V'\!}{x}$
  .
\end{proof}

Now we prove a key lemma, stating that parallel reduction $\Rightarrow$ coincides with strong parallel reduction $\Rrightarrow$ (the inclusion $\Rrightarrow \, \subseteq \, \Rightarrow$ holds trivially by definition of $\Rrightarrow$).
In its proof, as well as in the proof of Corollary~\ref{cor:postponeinternal} and Theorem~\ref{thm:trifurcate}, our Lemma~\ref{lemma:commutation} plays a crucial role: indeed, since head $\sigma$-reduction well interacts with head $\beta_v$-reduction, Takahashi's method \cite{Takahashi95} is still working when adding the reduction rules $\sigma_1$ and $\sigma_3$ to 
$\beta_v$-reduction.

\begin{lem}[Key Lemma]\label{lemma:main}
  If $M \Rightarrow M'$ then $M \Rrightarrow M'$.
\end{lem}
\begin{proof}
  By induction on the derivation of $M \Rightarrow M'$. Let us consider its last rule $\mathsf{r}$.
\begin{itemize}
  \item If $\mathsf{r} = \mathit{var}$ then $M = x \, M_1 \dots M_m$ and $M' = x \, M_1' \dots M_m'$ where $m \in \Nat$ and $M_i \Rightarrow M_i'$ for all $ 1 \leq i \leq m$. 
  By reflexivity of $\Rrightarrow$ (Lemma~\ref{rmk:reflexive}), $x \Rrightarrow x$. 
  By induction hypothesis, $M_i \Rrightarrow M_i'$ for all $1 \leq i\leq m$. Therefore, $M \Rrightarrow M'$ by applying Lemma~\ref{lemma:appstrong} $m$ times.

\item  If $\mathsf{r} = \lambda$ then $M = (\lambda x.M_0) M_1 \dots M_m$ and $M' = (\lambda x.M_0') M_1' \dots M_m'$ where $m \in \Nat$ and $M_i \Rightarrow M_i'$ for all $0 \leq i \leq m$. 
  By induction hypothesis, $M_i \Rrightarrow M_i'$ for all $1 \leq i \leq m$.
  According to Lemma~\ref{rmk:abs}.\ref{rmk:abs.parallel}, $\lambda x.M_0 \Rrightarrow \lambda x.M_0'$.
  So, $M \Rrightarrow M'$ by applying Lemma~\ref{lemma:appstrong} $m$ times.
  
\item  If $\mathsf{r} = \beta_v$ then $M = (\lambda x.M_0) V M_1 \dots M_m$ and $M' = M_0'\sub{V'\!}{x} M_1' \dots M_m'$ 
where $m \in \Nat$, $V \Rightarrow V'$ and $M_i \Rightarrow M_i'$ for all $0 \leq i \leq m$.
  By induction hypothesis, $V \Rrightarrow V'$ and $M_i \Rrightarrow M_i'$ for all $0 \leq i \leq m$.
  Moreover, $M_0\sub{V}{x} M_1 \dots M_m \Rrightarrow M'$ by Lemma~\ref{lemma:strongparallelsubstitution} and by applying Lemma~\ref{lemma:appstrong} $m$ times, thus $
  M_0\sub{V}{x}M_1 \!\dots M_m \ToHeadBetavReflTrans\! L \ToHeadSigmaReflTrans N \Inter M'$ for some $L, N \in \Lambda$.
  Therefore, $M \Rrightarrow M'\!$ since $M \ToHeadBetav\! M_0\sub{V}{x}M_1 \!\dots M_m$.
  
\item  If $\mathsf{r} = \sigma_1$ then $M = (\lambda x.M_0)N_0L_0 M_1 \dots M_m$ and $M' = (\lambda x.M_0'L_0')N_0' M_1' \dots M_m'$ where $m \in \Nat$, $L_0 \Rightarrow L_0'$, $N_0 \Rightarrow N_0'$ and $M_i \Rightarrow M_i'$ for any $0 \leq i \leq m$. 
  By induction hypothesis, $N_0 \Rrightarrow N_0'$ and $M_i \Rrightarrow M_i'$ for any $1 \leq i \leq m$.
  By applying the rule $\sigma_1$ for $\ToHeadSigma$, we have $M \ToHeadSigma (\lambda x.M_0L_0)N_0 M_1 \dots M_m$. By Lemma~\ref{rmk:preliminary}.\ref{rmk:preliminary.appparallel}, $M_0L_0 \Rightarrow M_0'L_0'$ and thus $\lambda x.M_0L_0 \Rrightarrow \lambda x.M_0'L_0'$ according to Lemma~\ref{rmk:abs}.\ref{rmk:abs.parallel}.
  So $(\lambda x.M_0L_0)N_0 M_1 \dots M_m \Rrightarrow M'$ by applying Lemma~\ref{lemma:appstrong} $m+1$ times, hence there are $L, N \!\in\! \Lambda$ such that $M \ToHeadSigma (\lambda x.M_0L_0)N_0 M_1 \dots M_m \allowbreak\ToHeadBetavReflTrans\! L \ToHeadSigmaReflTrans\! N \Inter M'$\!. By Lemma~\ref{lemma:commutation}.\ref{lemma:commutation.transitive}, there is $L' \!\in\! \Lambda$ such that $M \ToHeadBetavReflTrans L' \ToHeadSigmaReflTrans L \ToHeadSigmaReflTrans N \Inter M'$, so $M \Rrightarrow M'$.
\item  Finally, if $\mathsf{r} = \sigma_3$ then $M = V((\lambda x.L_0)N_0) N_1 \dots N_n$ and $M' \!= (\lambda x.V'L_0')N_0' N_1' \dots N_n'$ with $n \in \Nat$, $V \Rightarrow V'$, $L_0 \Rightarrow L_0'$ and $N_i \Rightarrow N_i'$ for any $0 \leq i \leq n$. 
  By induction hypothesis, $N_i \Rrightarrow N_i'$ for any $0 \leq i \leq n$.
  By the rule $\sigma_3$ for $\ToHeadSigma$, we have $M \ToHeadSigma (\lambda x.VL_0)N_0 N_1 \dots N_n$.
  By Lemma~\ref{rmk:preliminary}.\ref{rmk:preliminary.appparallel}, $VL_0 \Rightarrow V'L_0'$ and thus $\lambda x.VL_0 \allowbreak\Rrightarrow \lambda x.V'L_0'$ according to Lemma~\ref{rmk:abs}.\ref{rmk:abs.parallel}.
  So $(\lambda x.VL_0)N_0 N_1 \dots N_n \Rrightarrow M'$ by applying Lemma~\ref{lemma:appstrong} $n+1$ times, hence there are $L, N \!\in\! \Lambda$ such that $M \ToHeadSigma \allowbreak (\lambda x.VL_0)N_0 N_1 \dots N_n \allowbreak \ToHeadBetavReflTrans L \ToHeadSigmaReflTrans N \Inter M'$. 
  By Lemma~\ref{lemma:commutation}.\ref{lemma:commutation.transitive}, there is $L' \!\in\! \Lambda$ such that $M \ToHeadBetavReflTrans L' \ToHeadSigmaReflTrans L \ToHeadSigmaReflTrans N \Inter M'$, therefore $M \Rrightarrow M'$.
\qedhere
\end{itemize}
\end{proof}


\noindent Next Lemma~\ref{lemma:postpone} and Corollary~\ref{cor:postponeinternal} show that internal parallel reduction can be shifted after head $\vg$-reduction.

\begin{lem}[Postponement, version 1]\label{lemma:postpone}
  If $M \Inter L $ and $L \ToHeadBetav N$ (resp.\ $L \ToHeadSigma N$) then there exists $L' \in \Lambda$ such that $M \ToHeadBetav L'$ (resp.\ $M \ToHeadSigma L'$) and $L' \Rightarrow N$.
\end{lem}
\begin{proof}
  By induction on the derivation of $M \Inter L$. Let us consider its last rule $\mathsf{r}$.

  \begin{itemize}
  \item If $\mathsf{r} = \mathit{var}$, then $M = x = L$ which
    contradicts $L \ToHeadBetav N$ and $L \ToHeadSigma N$ by
    Lemmas~\ref{rmk:abs}.\ref{rmk:abs.novalue}-\ref{rmk:abs.notovalue}.

   \item If $\mathsf{r} = \lambda$ then $L = \lambda x.L'$ for some
    $L' \in \Lambda$, which contradicts $L \ToHeadBetav\! N$ and
    $L \ToHeadSigma\! N$ by
    Lemmas~\ref{rmk:abs}.\ref{rmk:abs.novalue}-\ref{rmk:abs.notovalue}.

  \item  Finally, if $\mathsf{r} = \mathit{right}$ then
    $M = VM_0 M_1 \dots M_m$ and $L = V'L_0 L_1 \dots L_m$ where
    $m \in \Nat$, $V \Rightarrow V'$ (so $V \Inter V'$ by
    Lemma~\ref{rmk:abs}.\ref{rmk:abs.value}), $M_0 \Inter L_0$ (thus
    $M_0 \Rightarrow L_0$ since $\Inter \, \subseteq \, \Rightarrow$)
    and $M_i \Rightarrow L_i$ for any $1 \leq i \leq m$.
    \begin{itemize}
    \item If $L \ToHeadBetav N$ then there are only two cases, depending on
      the last rule $\mathsf{r}'$ of the derivation of
      $L \ToHeadBetav N$.
      \begin{itemize}
      \item If $\mathsf{r}' = \beta_v$ then $V' = \lambda x.N_0'$,
        $L_0 \in \Lambda_v$ and $N = N_0'\sub{L_0}{x} L_1 \dots L_m$,
        thus $M_0 \in \Lambda_v$ and $V = \lambda x.N_0$ with
        $N_0 \Rightarrow N_0'$ by
        Lemma~\ref{rmk:abs}.\ref{rmk:abs.anti}.  By
        Lemma~\ref{lemma:parallelsubstitution}, 
        $N_0\sub{M_0}{x} \Rightarrow N_0'\sub{L_0}{x}$.  Let
        $L' = N_0\sub{M_0}{x} M_1 \dots M_m$: so
        $M = (\lambda x.N_0)M_0 M_1 \dots M_m \ToHeadBetav L'$ (apply
        the rule $\beta_v$ for $\ToHeadBetav$) and $L' \Rightarrow N$
        by applying
        Lemma~\ref{rmk:preliminary}.\ref{rmk:preliminary.appparallel}
        $m$ times.
      \item If $\mathsf{r}' = \mathit{right}$ then
        $N = V'N_0 L_1 \dots L_m$ with $L_0 \ToHeadBetav N_0$.  By
        induction hypothesis, there is $L_0' \in \Lambda$ such
        that $M_0 \ToHeadBetav L_0' \Rightarrow N_0$.  Let
        $L' = VL_0' M_1 \dots M_m$: so $M \ToHeadBetav L'$ (apply the
        rule $\mathit{right}$ for $\ToHeadBetav$) and
        $L' \Rightarrow N$ by applying
        Lemma~\ref{rmk:preliminary}.\ref{rmk:preliminary.appparallel}
        $m+1$ times.
      \end{itemize}

    \item If $L \ToHeadSigma N$ then there are only three cases, depending
      on the last rule $\mathsf{r}'$ of the derivation of
      $L \ToHeadSigma N$.
      \begin{itemize}
      \item If $\mathsf{r}' = \sigma_1$ then $m > 0$,
        $V' = \lambda x.N_0'$ and
        $N = (\lambda x.N_0'L_1)L_0 L_2 \dots L_m$, thus
        $V = \lambda x.N_0$ with $N_0 \Rightarrow N_0'$ by
        Lemma~\ref{rmk:abs}.\ref{rmk:abs.anti}.  Using
        Lemmas~\ref{rmk:preliminary}.\ref{rmk:preliminary.appparallel}
        and \ref{rmk:abs}.\ref{rmk:abs.parallel}, we have
        $\lambda x.N_0 M_1 \Rightarrow \lambda x.N_0'L_1$.  Let
        $L' = (\lambda x.N_0 M_1)M_0 M_2 \dots M_m$: so
        $M = (\lambda x.N_0)M_0 M_1 \dots M_m \ToHeadSigma L'$ (apply
        the rule $\sigma_1$ for $\ToHeadSigma$) and $L' \Rightarrow N$
        by applying
        Lemma~\ref{rmk:preliminary}.\ref{rmk:preliminary.appparallel}
        $m$ times.

      \item If $\mathsf{r}' = \sigma_3$ then
        $L_0 = (\lambda x.L_{01})L_{02}$ and
        $N = (\lambda x.V'L_{01})L_{02} L_1 \dots L_m$.  Since
        $M_0 \Inter (\lambda x.L_{01})L_{02}$, then by simple inspection of the rules for $\Inter$ (Definition~\ref{def:internal}) we infer that $M_0 = V_0M_{02}$ with  $V_0 \Rightarrow \lambda x.L_{01}$ and $M_{02} \Inter L_{02}$ (so $M_{02} \Rightarrow L_{02}$ because $\Inter \, \subseteq \, \Rightarrow$).  
        By Lemmas~\ref{rmk:abs}.\ref{rmk:abs.value}~and~\ref{rmk:abs}.\ref{rmk:abs.anti}, from $V_0 \Rightarrow \lambda x.L_{01}$ it follows that $V_0 = \lambda x.M_{01}$ with $M_{01} \Rightarrow L_{01}$.
	By Lemmas~\ref{rmk:preliminary}.\ref{rmk:preliminary.appparallel}
        and \ref{rmk:abs}.\ref{rmk:abs.parallel}, 
        $\lambda x.V M_{01} \Rightarrow \lambda x.V'L_{01}$.  Let
        $L' = (\lambda x.V M_{01})M_{02} M_1 \dots M_m$: so
        $M = V((\lambda x.M_{01})M_{02}) M_1 \dots M_m \ToHeadSigma
        L'$
        (apply the rule $\sigma_3$ for $\ToHeadSigma$) and
        $L' \Rightarrow N$ by applying
        Lemma~\ref{rmk:preliminary}.\ref{rmk:preliminary.appparallel}
        $m+1$ times.

      \item If $\mathsf{r}' = \mathit{right}$ then
        $N = V'N_0 L_1 \dots L_m$ with $L_0 \ToHeadSigma N_0$. By
        induction hypothesis, there exists $L_0' \in \Lambda$ such
        that $M_0 \ToHeadSigma L_0' \Rightarrow N_0$.  Let
        $L' = VL_0' M_1 \dots M_m$: so $M \ToHeadSigma L'$ (apply the
        rule $\mathit{right}$ for $\ToHeadSigma$) and
        $L' \Rightarrow N$ by applying
        Lemma~\ref{rmk:preliminary}.\ref{rmk:preliminary.appparallel}
        $m+1$ times.
      \qedhere
      \end{itemize}
    \end{itemize}
  \end{itemize}

\end{proof}

\begin{cor}[Postponement, version 2]
\label{cor:postponeinternal}
  If $M \Inter L $ and $L \ToHeadBetav N$ (resp.~$L \ToHeadSigma N$), then there exist $L',L'' \in \Lambda$ such that $M \ToHeadBetavTrans L' \ToHeadSigmaReflTrans L'' \Inter N$ (resp.~$M \ToHeadBetavReflTrans L' \ToHeadSigmaReflTrans L'' \Inter N$).
\end{cor}
\begin{proof}
  Immediate from Lemmas~\ref{lemma:postpone} and~\ref{lemma:main}, applying Lemma~\ref{lemma:commutation}.\ref{lemma:commutation.transitive} if $L \ToHeadSigma N$.
\end{proof}


\proof[{\bf Proof of Sequentialization (Theorem \ref{thm:trifurcate} on page \pageref{thm:trifurcate})}]$\;$\\
  By Lemma~\ref{rmk:preliminary}.\ref{rmk:preliminary.preConfluence}, 
  $M \!\Rightarrow^*\! M'$ and thus there are $m \in \Nat$ and $M_0, \dots, M_m \in \Lambda$ such that $M = M_0$, $M_m = M'$ and $M_i \Rightarrow M_{i+1}$ for any $0 \leq i < m$. 
  We prove by induction on $m \in \Nat$ that there are $L,N \in \Lambda$ such that $M \ToHeadBetavReflTrans L \ToHeadSigmaReflTrans N \!\InterTrans\! M'$, so $N \ToIntTrans M'$ by Lemma~\ref{rmk:preliminary}.\ref{rmk:preliminary.interinclusion}. 
  
  \begin{itemize}
  \item If $m = 0$ then $M = M_0 = M'$ and hence we conclude by taking
    $L = M' = N$.

  \item  Suppose $m > 0$. By induction hypothesis applied to
    $M_1 \Rightarrow^* M'$, there are $L', N' \in \Lambda$ such that
    $M_1 \ToHeadBetavReflTrans L' \ToHeadSigmaReflTrans N' \InterTrans
    M'$.
    By applying Lemma~\ref{lemma:main} to $M \Rightarrow M_{1}$, there exist
    $L_0, N_0 \in \Lambda$ such that
    $M \ToHeadBetavReflTrans L_0 \ToHeadSigmaReflTrans N_0 \Inter
    M_1$.
    By applying Corollary~\ref{cor:postponeinternal} repeatedly, there
    is $N \in \Lambda$ such that
    $N_0 \ToHeadReflTrans N \Inter N'$, and hence
    $M \ToHeadReflTrans N \InterTrans M'$.  According to
    Lemma~\ref{lemma:commutation}.\ref{lemma:commutation.general},
    there is $L \in \Lambda$ such that
    $M \ToHeadBetavReflTrans L \ToHeadSigmaReflTrans N \InterTrans
    M'$.  \qed
  \end{itemize}

\section{Standardization}\label{sect:sos}

This section is devoted to prove the 
standardization theorem for $\lambda_v^\sigma$, stating that if $N \to_\vg^* L$ then there is a ``standard'' $\vg$-reduction sequence from $N$ to $L$ (Theorem~\ref{thm:standardization}).
Roughly speaking, a reduction sequence is \emph{standard} if the ``positions'' of the 
reduced redexes move from left to right.%
\footnote{In ordinary $\lambda$-calculus, standard sequences (for $\beta$-reduction) can be described as follows: ``After each contraction of a redex $R$, index the $\lambda$'s of redexes to the left of $R$. Redexes with indexed $\lambda$'s are not allowed to be contracted anymore. Indexed $\lambda$'s remain indexed after contractions
of other redexes'' \cite[p.\,297]{barendregt84nh}.}
Actually, in a call-by-value $\lambda$-calculus (such as $\lambda_v^\sigma$), this ``left-to-right'' order is more delicate to define, since $\beta$-redexes can be fired 
only after their arguments have been reduced to a value,%
\footnote{E.g., according to \cite{barendregt84nh,Crary09} (and us), in $\lambda_v$ the $\beta_v$-reduction sequence $\Delta (II) \ToHeadBetav \Delta I \ToHeadBetav II$ is standard, even if the (only) $\beta_v$-redex 
in $\Delta(II)$ seems to be ``on the right'' of the $\beta_v$-redex $\Delta I$ 
reduced later. 
The subtlety is that in $\lambda_v$, 
unlike $\lambda$,  new redexes can be created in the following way: 
a $\beta_v$-reduction step in the argument of a $\beta$-(not $\beta_v$-)redex $R$ may turn the argument itself into a value, turning 
$R$ into a $\beta_v$-redex.}
but the essence is the same: a standard reduction sequence begins with head reduction steps, and then continues with internal reduction steps selecting redexes 
according to a ``left-to-right'' order.
Our choice to prioritize head reduction over internal reduction (followed also by \cite{Plotkin75,Crary09} for $\lambda_v$) entails that in a standard sequence some 
changes of positions from right to left for the selected redexes
may take place when passing from the head reduction phase to the internal reduction one: e.g. according to \cite{Plotkin75,Crary09} (and us), in $\lambda_v$ the sequence $(\lambda x. I\Delta)(y(II)) \ToHeadBetav (\lambda x.I\Delta) (zI) 
\to_{\beta_v} (\lambda x.\Delta)(zI)$ is standard, even if the $\beta_v$-redex $II$ is ``on the right'' of the $\beta_v$-redex $I\Delta$ (fired later as internal, since it is under the scope of a $\lambda$).
Actually, in $\lambda_v^\sigma$ another intricacy arises in defining a standard order: there are not only $\beta$-redexes but also $\sigma$-redexes and they may overlap. Our approach is to prioritize head $\beta_v$-redexes over head $\sigma$-redexes 
(this idea extends iteratively to subterms).

We 
define the notion of standard reduction sequence by closely following the 
approach used in \cite{Plotkin75,Crary09}%
{, so the redex-order is defined by induction on the structure of terms, without involving any (tricky) notion of residual redex.} 




\begin{defi}[Standard head sequence]
\label{def:head-seq}
  For any $k, m \in \Nat$ with $k \leq m$, 
  a \emph{standard head sequence}, denoted by $\sequence{M_0 ,\dots,M_k, \dots, M_m}{\head}$, is a finite sequence $(M_0, \dots, M_k, \dots, M_m)$ of terms
  such that $M_i \ToHeadBetav M_{i+1}$ for any $0 \leq i < k$, and
  $M_i \ToHeadSigma M_{i+1}$ for any $k\leq i < m$.
\end{defi}
In other words, a standard head sequence is a head $\vg$-reduction sequence where the head $\beta_v$-reduction steps precede all the head $\sigma$-reduction steps, without any order between head $\sigma_1$- and head $\sigma_3$-reduction steps.
Note that when $k = m$ (resp.~$k = 0$) then the standard head sequence consists only of head $\beta_v$-(resp.~head $\sigma$-)reduction steps.
It is easy to check that $\sequence{M}{\head}$ for every 
$M \in \Lambda$ (apply Definition~\ref{def:head-seq} with $m = 0$).

Using the above definition of standard head sequence, we define by mutual induction the notions of standard sequence and standard inner sequence of terms (Definition~\ref{def:std}). 

\begin{defi}[Standard and standard inner sequences]\label{def:std}
\emph{Standard sequences} and \emph{standard} \emph{inner sequences} of terms, denoted by $\sequence{M_0 ,\dots, M_m}{\mathit{std}}$ and $\sequence{M_0 ,\dots, M_m}{\mathit{in}}$ respectively (with $m \in \Nat$ and $M_0, \dots, M_m \in \Lambda$), are defined by mutual induction as follows:
\begin{enumerate}
  \item \label{1} if $\sequence{M_0 ,\dots, M_m}{\head}$ and $\sequence{M_m,\dots, M_{m+n}}{\mathit{in}}$, then  
  $\sequence{M_0 ,\dots, M_m, \dots, M_{m+n}}{\mathit{std}}$
  ;  
  \item \label{2} $\sequence{M}{\mathit{in}}$, for every $M \in \Lambda$;
  \item \label{3} if $\sequence{M_0 ,\dots, M_m}{\mathit{std}}$  
then $\sequence{\lambda z.M_0, \dots, \lambda z.M_m}{\mathit{in}}$
  ;
  \item \label{4} if  $\sequence{V_0, \dots, V_h}{\mathit{std}}$ 
  and $\sequence{M_0, \dots, M_m}{\mathit{in}}$, then 
  $\sequence{V_0M_0, \dots, V_h M_0, \allowbreak\dots, V_h M_m}{\mathit{in}}$ 
  (where $V_0, \dots, V_h \in \Lambda_v$); 
  \item \label{5} if  $\sequence{M_0, \dots, M_m}{\mathit{in}}$, $\sequence{L_0 ,\dots, L_l}{\mathit{std}}$  and $M_0\not\in\Lambda_v$, then 
$\sequence{M_0 L_0, \dots, M_m L_0, \dots ,M_m L_l}{\mathit{in}}$
.
  \end{enumerate}
\end{defi}

\begin{rem}\label{rmk:standard-to-reduction}
  It is easy to show (by mutual induction on the definition of standard and standard inner sequences) that, given $n \in \Nat$ and $N_0, \dots, N_n \in \Lambda$, if $\sequence{N_0, \dots, N_n}{\mathit{in}}$ (resp.~$\sequence{N_0, \dots, N_n}{\mathit{std}}$) then $N_i \ToInt N_{i+1}$ (resp.~$N_i \to_\vg N_{i+1}$) for any $0 \leq i < n$.
\end{rem}

In fact, the presence of a standard or standard inner sequence means that not only there is a $\vg$-reduction sequence or an internal $\vg$-reduction sequence, respectively, but also that this $\vg$-reduction sequence is performed 
selecting $\vg$-redexes according to the aforementioned ``left-to-right'' order, up to some intricacies already pointed out on p.~\pageref{def:head-seq}
.
Indeed, in Definition~\ref{def:std}, the rule \eqref{1}\,---\,the only one yielding 
standard sequences\,---\,says that standard sequences start by 
reducing first head $\beta_v$-redexes, then head $\sigma$-redexes and then internal $\vg$-redexes, where the head $\beta_v$-redex in a term is its (unique, if any) \emph{leftmost-outermost} $\beta_v$-redex not under the scope of $\lambda$'s, and head $\sigma$-redexes in a term are its (possibly not unique) \emph{leftmost-outermost} $\sigma$-redexes not under the scope of $\lambda$'s.
Rules \eqref{4}-\eqref{5} in Definition~\ref{def:std} intuitively mean that the positions of the $\vg$-redexes 
reduced in a standard inner sequence move from left to right.

In order to give informative examples about standard and standard inner sequences
, for any $\Rule \in \{\beta_v, \sigma_1, \sigma_3, \sigma\}$ we set 
$\ToIntRule \,=\, \to_{\Rule} \cap \ToInt$.
 
\begin{exa}
  Let $L = (\lambda y.Ix)(z (\Delta I))(II)$: one has that
  $L \!\ToIntBetav\! (\lambda y.Ix)(z(\Delta I))I \!\ToIntBetav\! (\lambda y.x)(z(\Delta I))I$ and $L \ToHeadSigmaOne (\lambda y.Ix(II))(z(\Delta I)) \ToHeadBetav (\lambda y.Ix(II))(z(II))$ 
  are not standard sequences; 
  but $L \!\ToIntBetav\! (\lambda y.Ix)(z(\Delta I))I$ and $L \ToHeadBetav (\lambda y.Ix)(z(II))(II) \ToHeadBetav (\lambda y.Ix)(zI)(II) \allowbreak\ToHeadSigmaOne (\lambda y.Ix(II))(zI) \allowbreak\ToIntBetav (\lambda y.x(II))(zI) \ToIntBetav (\lambda y.xI)(zI)$ 
  are standard sequences.
\end{exa}

    

The next lemma states that standard head and inner sequences are standard sequences.

\begin{lem}\label{lemma:inner-standard} 
  Given $n \in \Nat$, if $\sequence{N_0 ,\dots, N_n}{in}$ (resp.\ $\sequence{N_0 , \dots, N_n}{\head}$) then
  $\sequence{N_0 ,\dots, N_n}{std}$.
\end{lem}
\begin{proof}
  $\sequence{N_0}{\head}$ (resp.\ $\sequence{N_n}{in}$ by Definition~\ref{def:std}.\ref{2}), 
so $\sequence{N_0 ,\dots, N_n}{std}$ by Definition~\ref{def:std}.\ref{1}.
\end{proof}

In particular, $\sequence{N}{\mathit{std}}$ for any $N \in \Lambda$: apply Definition~\ref{def:std}.\ref{2} and Lemma~\ref{lemma:inner-standard} for $n = 0$.

Note that the concatenation of two standard sequences is not standard, in general: take for instance a standard inner sequence followed by a standard head sequence.

For all $V_0, \dots, V_h \in \Lambda_v$, $\sequence{V_0, \allowbreak\dots, V_h}{\mathit{std}}$ iff $\sequence{V_0, \dots, V_h}{\mathit{in}}$: 
the left-to-right implication follows from 
the rule \eqref{1} of Definition~\ref{def:std} (the only one yielding standard sequences) and Remarks~\ref{rmk:abs}.\ref{rmk:abs.novalue}-\ref{rmk:abs.notovalue} ($\sequence{V_0, \dots, V_h}{\head}$ is impossible for $h >0$); 
the converse holds by Lemma~\ref{lemma:inner-standard}.

We can now state and prove the standardization theorem for $\lambda_v^\sigma$, one of the main result of this paper: if $M$ $\vg$-reduces to $M'$ then there exists a standard sequence from $M$ to $M'$.
The idea to build this standard sequence is to sequentialize (as stated in Theorem~\ref{thm:trifurcate}) the $\vg$-reduction sequence from $M$ to $M'$ iteratively according to a ``left-to-right'' order.

\begin{thm}[Standardization] Let $M$ and $M'$ be terms.
\label{thm:standardization}
  \begin{enumerate}
    \item\label{thm:standardization.head} If $M \ToHeadReflTrans M' $ then there is a standard head sequence $\sequence{M ,\dots, M'}{\head}$.
    \item \label{thm:standardization.inner} If $M \ToIntTrans M' $ then there is a standard inner sequence $\sequence{M ,\dots, M'}{\mathit{in}}$.
    \item \label{thm:standardization.standard} If $M  \to_\vg^*  M'$ then there is a standard sequence $\sequence{M ,\dots, M'}{\mathit{std}}$.
  \end{enumerate}
\end{thm}
\begin{proof}
Theorem~\ref{thm:standardization}.\ref{thm:standardization.standard} is an immediate consequence of Theorems~\ref{thm:standardization}.\ref{thm:standardization.head}-\ref{thm:standardization.inner} and Theorem~\ref{thm:trifurcate}: indeed, if $M \to_\vg^* M'$ then there is a term $M''$ such that $M \ToHeadReflTrans M'' \ToIntTrans M'$ by sequentialization (Theorem~\ref{thm:trifurcate}), moreover $M \ToHeadReflTrans M''$ implies that there is a sequence $\sequence{M ,\dots, M''}{\head}$ by Theorem~\ref{thm:standardization}.\ref{thm:standardization.head}, and $M'' \ToIntTrans M'$ implies that there is a sequence $\sequence{M'', \dots, M'}{in}$ by Theorem~\ref{thm:standardization}.\ref{thm:standardization.inner}.
According to the rule \eqref{1} of Definition~\ref{def:std}, $\sequence{M, \dots, M'', \dots, M'}{std}$.

It remains to prove Theorems~\ref{thm:standardization}.\ref{thm:standardization.head}-\ref{thm:standardization.inner}.
Now, Theorem~\ref{thm:standardization}.\ref{thm:standardization.head} is exactly our Lemma~\ref{lemma:commutation}.\ref{lemma:commutation.general}, already proved.

Theorem~\ref{thm:standardization}.\ref{thm:standardization.inner} is proved by induction on $M' \in \Lambda$, using Theorem~\ref{thm:standardization}.\ref{thm:standardization.head}.
  \begin{itemize}
    \item If  $M'= z$ then $M=z$ by Lemma
    \ref{rmk:preliminary}.\ref{rmk:preliminary.interExpansion}, thus $\sequence{z}{in}$ by the rule \eqref{2} of Definition~\ref{def:std}.
    
    \item If $M'= \lambda z.L'$ then there is $L \in \Lambda$ such that $M =\lambda z.L$ and $L \to_\vg^* L'$, by 
    Lemma~\ref{rmk:preliminary}.\ref{rmk:preliminary.interExpansion}. 
    By sequentialization (Theorem~\ref{thm:trifurcate}), there exists a term $N$ such that $L \ToHeadReflTrans N \ToIntTrans L'$.
    By Theorem~\ref{thm:standardization}.\ref{thm:standardization.head}, from $L \ToHeadReflTrans N$ it follows that there is a sequence $\sequence{L, \dots, N}{\head}$.
    By induction hypothesis applied to $N \ToIntTrans L'$, there is a sequence $\sequence{N,\dots,L'}{\mathit{in}}$.
    According to the rule \eqref{1} of Definition~\ref{def:std}, $\sequence{L,\dots, N, \dots,L'}{\mathit{std}}$.
    By the rule \eqref{3} of Definition~\ref{def:std}, $\sequence{\lambda  z.L , \dots, \lambda z.N, \dots, \lambda z.L'}{in}$, that is $\sequence{M, \dots, M'}{in}$.
    
    \item If $M'=N'L'$ then $M = NL$ for some $N, L \in \Lambda$ by Remark~\ref{rmk:fromvalue}, since $
    \ToIntTrans \,\subseteq\, \to_\vg^*$ and $M' \notin \Lambda_v$.
    By Lemma~\ref{rmk:preliminary}.\ref{rmk:preliminary.interinclusion}, 
    $NL \InterTrans N'L'$; 
    clearly, for each step of $\Inter$ in $NL \InterTrans N'L'$, the last rule of its derivation is an instance of the rule $\mathit{right}$ for $\Inter$ (the other rules deal with values, see Definition \ref{def:internal}).
    There are two sub-cases.
    \begin{itemize}
      \item If $N \in \Lambda_v$ then $N \Rightarrow^* N'$ and $L \InterTrans L'$, so $N \to_\vg^* N'$ and $L \ToIntTrans L'$ by Lemmas~\ref{rmk:preliminary}.\ref{rmk:preliminary.preConfluence}-\ref{rmk:preliminary.interinclusion}.
      By sequentialization (Theorem~\ref{thm:trifurcate}), there is a term $N''$ such that $N \ToHeadReflTrans N'' \ToIntTrans N'$, and actually $N = N''$ by Lemmas~\ref{rmk:abs}.\ref{rmk:abs.novalue}-\ref{rmk:abs.notovalue} since $N$ is a value; thus, $N \ToIntTrans N'$.
      By induction hypothesis applied to $N \ToIntTrans N'$ and $L \ToIntTrans L'$, there are sequences $\sequence{N, \dots, N'}{in}$ (hence $\sequence{N, \dots, N'}{std}$ by Lemma~\ref{lemma:inner-standard}) and $\sequence{L,\dots,L'}{\mathit{in}}$.
      In particular, according to Remark~\ref{rmk:standard-to-reduction}, if $\sequence{N, \dots, N'}{std} = (N_0, \dots, N_n)$ for some $n \in \Nat$ and $N_0, \dots, N_n \in \Lambda$ (with $N_0 = N$ and $N_n = N'$), then $N_i \to_\vg N_{i+1}$ for all $0 \leq i < n$, and hence $N_0, \dots, N_n$ (that is, all the terms in $\sequence{N, \dots, N'}{std}$) are values by Remark~\ref{rmk:fromvalue}, since $N_0$ is a value.
      By applying the rule \eqref{4} of Definition~\ref{def:std}, 
      $\sequence{NL,\dots,N'L, \dots, N' L' }{in}$. 
	
      \item If $N \notin \Lambda_v$ (i.e.~$N = V\!M_{1} \dots M_{m}$ with $m > 0$, by Remark~\ref{rem:shape}) then $N \InterTrans\! N'$ and $L \Rightarrow^*\! L'$, so $N \ToIntTrans\! N'$ and $L \to_\vg^* L'$ 
      by Lemmas~\ref{rmk:preliminary}.\ref{rmk:preliminary.preConfluence}-\ref{rmk:preliminary.interinclusion}. 
      By sequentialization (Theorem~\ref{thm:trifurcate}), $L \ToHeadReflTrans L'' \ToIntTrans L'$ for some term $L''$.
      By Theorem~\ref{thm:standardization}.\ref{thm:standardization.head}, there is a sequence $\sequence{L, \dots, L''}{\head}$.
      By induction hypothesis applied to $N \ToIntTrans N'$ and $L'' \ToIntTrans L'$, there are sequences $\sequence{N, \dots, N'}{in}$ and $\sequence{L'',\dots,L'}{\mathit{in}}$.
      According to the rule \eqref{1} of Definition~\ref{def:std}, $\sequence{L ,\dots, L'', \dots, L'}{std}$.
      By applying the rule \eqref{5} of Definition~\ref{def:std}, 
      $\sequence{NL ,\dots,N' L ,\dots, N' L'}{in}$, that is $\sequence{M, \dots, M'}{in}$.
    \qedhere
    \end{itemize}
  \end{itemize}
\end{proof}


\noindent Theorem~\ref{thm:standardization} gives only a \emph{weak standardization}: it rearranges a $\vg$-reduction sequence from $M$ to $M'$ so as to obtain a standard sequence 
from $M$ to $M'$, but a standard sequence selects $\vg$-redexes following a \emph{partial} (and not total, in general) \emph{order} on $\vg$-redexes.
Indeed, a standard sequence is not uniquely determined by its starting and end terms, 
and this is essentially due to 
two facts (exemplified by Examples~\ref{ex:sigma-no-total}-\ref{ex:betav-sigma-no-total}
, respectively):
\begin{enumerate}
  \item as already remarked on pp.~\pageref{def:headsigma}-\pageref{fig:sigma-overlap}, head $\sigma$-redexes may overlap and be incomparable;
  \item in a standard (head) sequence, there is no restriction on when ending a head $\beta_v$-reduction phase and beginning a head $\sigma$-reduction phase.
\end{enumerate}

\begin{exa} 
\label{ex:sigma-no-total}
  The following $\sigma$-reduction sequences (fired $\sigma$-redexes are underlined)
  \begin{align*}
    \underline{I (\Delta I) I} &\ToHeadSigmaOne \underline{(\lambda x. xI) (\Delta I)} \ToHeadSigmaThree (\lambda z. (\lambda x. xI) (zz)) I \qquad \textup{and}\\
    \underline{I (\Delta I}) I  &\ToHeadSigmaThree \underline{(\lambda z. I (zz))II} \ToHeadSigmaOne (\lambda z. \underline{I (zz) I})I \ToIntSigmaOne (\lambda z. (\lambda x. xI) (zz)) I
  \end{align*}
 are both\,---\,different\,---\,standard sequences from $I (\Delta I) I$ to $(\lambda z. (\lambda x. xI) (zz)) I$.
\end{exa}

\begin{exa}
  \label{ex:betav-sigma-no-total}
  The following head $\vg$-reduction sequences (fired $\vg$-redexes are underlined)
  \begin{align*}
    I( \underline{\Delta\Delta}) I &\ToHeadBetav \underline{I(\Delta\Delta) I} \ToHeadSigmaOne (\lambda x.x I)(\Delta\Delta) & &\textup{and} & 
   \underline{I(\Delta\Delta) I} &\ToHeadSigmaOne (\lambda x.x I)(\Delta\Delta)
  \end{align*}
   are both\,---\,different\,---\,standard (head) sequences from $I (\Delta \Delta) I$ to $(\lambda x. xI) (\Delta\Delta)$.
\end{exa}


Finally, we compare our notion of standardization with that for Plotkin's $\lambda_v$ given in \cite[p.~137]{Plotkin75} and \cite{Crary09}.
To make the comparison possible we neglect $\sigma$-reduction and we recall that
$\ToHeadBetav$ is exactly Plotkin's left-reduction \cite[p.~136]{Plotkin75}. 
As remarked in \cite[
p.~149]{herbelin09lncs}, 
both 
$(\lambda z.II) (II) \ToIntBetav (\lambda z.I) (II)  \ToHeadBetav (\lambda z.I) I$ 
and
$(\lambda z.II) (II) \ToHeadBetav (\lambda z.II) I  \ToIntBetav (\lambda z.I) I$
are standard sequences from $(\lambda z.II) (II)$ to $(\lambda z.I)I$ according to \cite{Plotkin75,Crary09}.
However, only the second sequence is standard in our sense 
{(our standardization restricted to $\to_{\beta_v}$ is exactly the parametric standardization of \cite{paolini04iandc} for 
$\lambda_v$, which imposes a total order on $\beta_v$-redexes).}
Without the distinction in Definition~\ref{def:std} between standard and standard inner sequences,
both the above sequences would be standard; indeed, \cite{Plotkin75,Crary09} do not make this distinction and their standardization imposes only a partial order on $\beta_v$-redexes.


\section{Conservativity}\label{sect:conservative}

We now present our main contribution: the shuffling calculus $\lambda_v^\sigma$ is a \emph{conservative} extension of $\lambda_v$.
To be precise, we will prove 
that $\lambda_v^\sigma$ is sound with respect to the 
observational equivalence introduced by Plotkin in \cite{Plotkin75} for $\lambda_v$ (Corollary~\ref{cor:observational}),  
and that 
the 
notions of potential valuability and solvability  for $\lambda_v$, introduced in \cite{PaoliniRonchi99}, 
coincide with the respective notions for $\lambda_v^\sigma$ (Theorem~\ref{thm:valsolv}). 
This justifies the idea that $\lambda_v^\sigma$ is a useful tool for studying 
properties of $\lambda_v$, as stated in \cite{carraro14lncs}. 
All these results can be 
proved using standardization for $\lambda_v^\sigma$. 
Actually, the following corollary of sequentialization 
(Theorem~\ref{thm:trifurcate}) is enough.


\begin{cor}[Reduction to a value]
\label{cor:value}
Let $M \in \Lambda$ and $V \in \Lambda_v$.
\begin{enumerate}
  \item\label{cor:value.one} If $M \to_{\vg}^* V$ then there exists $V' \in \Lambda_v$ such that $M \ToHeadBetavReflTrans V' 
  \ToIntTrans V$.
  \item\label{cor:value.two} 
  $M \ToHeadBetavReflTrans V$ if and only if $M \ToHeadReflTrans V$.
\end{enumerate}
\end{cor}
\begin{proof}\hfill
  \begin{enumerate}
    \item By sequentialization (Theorem~\ref{thm:trifurcate}), $M \ToHeadBetavReflTrans L \ToHeadSigmaReflTrans N \ToIntTrans V$ for some $N, L \in \Lambda$.
    By Lemma~\ref{rmk:preliminary}.\ref{rmk:preliminary.interExpansion}, $N \in \Lambda_v$ and thus $L = N$ according to Lemma~\ref{rmk:abs}.\ref{rmk:abs.notovalue}.

    \item $\Leftarrow$: By Lemma~\ref{lemma:commutation}.\ref{lemma:commutation.general}, $M \ToHeadBetavReflTrans N \ToHeadSigmaReflTrans V$ for some $N$, and $N = V$ by Lemma~\ref{rmk:abs}.\ref{rmk:abs.notovalue}. 
    
    \noindent $\Rightarrow$: Trivial, since $\ToHeadBetav \, \subseteq \, \ToHead$.
    \qedhere
  \end{enumerate}
\end{proof}

\noindent Corollary~\ref{cor:value} gives a first conservative result of $\lambda_v^\sigma$ with respect to $\lambda_v$: roughly
, it says that if a term $M$ $\vg$-reduces to a value (this is the case in particular for all closed head $\vg$-normalizable terms, as proven in \cite{guerrieri15wpte}) then $\sigma$-reduction steps are ``useless'' since head $\beta_v$-reduction\,---\,i.e.~Plotkin's evaluation for $\lambda_v$\,---\,reduces $M$ to a value (Corollary~\ref{cor:value}.\ref{cor:value.one}) and this value is the same as the one reached by means of head $\vg$-reduction 
(Corollary~\ref{cor:value}.\ref{cor:value.two}).

\begin{rem}[Uniqueness of head $\vg$-normal forms that are values]
\label{rmk:unique}
  Incidentally, notice that, in spite of the non-confluence of head $\vg$-reduction shown in Figure~\ref{fig:sigma-overlap}, Corollary~\ref{cor:value}.\ref{cor:value.two} entails that if $M \ToHeadReflTrans V \in \Lambda_v$ then $V$ is the unique head $\vg$-normal form of $M$. 
  Indeed, let $N$ be a head $\vg$-normal form of $M$ (namely, $M \ToHeadReflTrans N$ and $N$ is head $\vg$-normal): by confluence of $\to_\vg$ (Proposition~\ref{prop:general-properties}
  ), there is a term $L$ such that $V \to_\vg^* L \,\, {}_\vg^*\!\!\leftarrow N$, in particular $N \ToIntTrans L$ because $N$ is head $\vg$-normal; 
  by Remark~\ref{rmk:fromvalue}, $L \in \Lambda_v$ (since $V$ is a value) and then $N \in \Lambda_v$ by Lemma~\ref{rmk:preliminary}.\ref{rmk:preliminary.interExpansion};
  according to Corollary~\ref{cor:value}.\ref{cor:value.two}, $M \ToHeadBetavReflTrans N$ and hence $N = V$, since head $\beta_v$-reduction is deterministic and values are head $\beta_v$-normal (Lemma~\ref{rmk:abs}.\ref{rmk:abs.novalue}).
  More details about terms having a unique head $\vg$-normal form are in \cite{guerrieri15wpte}.
\end{rem}
 
Let us recall the notion of 
observational 
equivalence introduced by Plotkin \cite{Plotkin75} for $\lambda_v$.
Informally, two terms are observationally equivalent if they can be substituted for each other in all contexts without observing any difference in their behaviour, where ``behaviour'' 
means to test if call-by-value evaluation (head $\beta_v$-reduction) terminates on a value or not.

\begin{defi}[Halting, observational equivalence]\label{def:halt}
  Let $M \in \Lambda$.
  \begin{itemize}
  \item We say that (\emph{the evaluation of}) \emph{$M$ halts} if
    there exists $V \in \Lambda_v$ such that $M \ToHeadBetavReflTrans
    V$.
\item    The (\emph{call-by-value}) \emph{observational equivalence} is an
    equivalence relation $\cong$ on $\Lambda$ defined by: $M \cong N$
    if, for every context $\mathtt{C}$, one has that $\ctxC{M}$ halts
    iff $\ctxC{N}$ halts.%
  \end{itemize}
\end{defi}

\noindent Plotkin's original definition of call-by-value observational equivalence 
\cite[p.~144]{Plotkin75} 
also requires that $\ctxC{M}$ and $\ctxC{N}$ are closed terms, according to the tradition identifying programs with closed terms. However, the two equivalences coincide.

Clearly, the notions of halting and observational equivalence can be defined also for $\lambda_v^\sigma$, using 
$\ToHead$ instead of $\ToHeadBetav$ in Definition~\ref{def:halt}.
But head $\sigma$-reduction plays no role neither in deciding the halting problem for evaluation (Corollary~\ref{cor:value}.\ref{cor:value.one}), nor in reaching a particular value (Corollary~\ref{cor:value}.\ref{cor:value.two}). 
Therefore, we can conclude that the notions of halting and observational equivalence in $\lambda_v^\sigma$ 
\emph{coincide} with those in $\lambda_v$, respectively.

Now we compare the equational theory of $\lambda_v^\sigma$ with Plotkin's observational equivalence.

\begin{thm}[Adequacy of $\vg$-reduction]\label{thm:adequacy}
    If $M \to_{\vg}^* M'$ then: $M$ halts iff $M'$ halts.
\end{thm}

\begin{proof}
  If $M'$ halts then $M' \ToHeadBetavReflTrans V \in \Lambda_v$ and hence $M \to_{\vg}^* M' \to_{\vg}^* V$ since $\ToHeadBetav \, \subseteq \, \to_{\vg}$. By Corollary~\ref{cor:value}.\ref{cor:value.one}, there exists $V' \in \Lambda_v$ such that $M \ToHeadBetavReflTrans V'$. Thus, $M$ halts.
  
  Conversely, if $M$ halts then $M \ToHeadBetavReflTrans V \in \Lambda_v$, so $M \to_{\vg}^* V$ since $\ToHeadBetav \, \subseteq \, \to_{\vg}$. 
  By confluence of $\to_\vg$ (Proposition~\ref{prop:general-properties}, since $M \!\to_{\vg}^*\! M'$) and Remark~\ref{rmk:fromvalue} (as $V \!\in\! \Lambda_v$), $V \to_{\vg}^* V'\!$ and $M' \to_{\vg}^* V'$ for some $V' \in \Lambda_v$. 
  By Corollary~\ref{cor:value}.\ref{cor:value.one}, $M' \ToHeadBetavReflTrans V''$ for some $V'' \!\in\! \Lambda_v$. Therefore, $M'$ halts.
\end{proof}



\begin{cor}[Soundness with respect to $\lambda_v$]\label{cor:observational}
  If $M =_\vg N$ then $M \cong N$.
\end{cor}
\begin{proof}
  Let $\mathtt{C}$ be a context. 
  By confluence of $\to_\vg$ (Proposition~\ref{prop:general-properties}), $M =_\vg N$ implies that there exists $L\in \Lambda$ such that $M \to_\vg^* L$ and $N \to_\vg^* L$, hence $\ctxC{M} \to_\vg^* \ctxC{L}$ and $\ctxC{N} \to_\vg^* \ctxC{L}$. 
  By Theorem~\ref{thm:adequacy}, $\ctxC{M}$ halts iff $\ctxC{L}$ halts iff $\ctxC{N}$ halts. Therefore, $M \cong N$.
\end{proof}

Plotkin \cite[p.~144]{Plotkin75} has already proved that $M =_{\beta_v} N$ implies $M \cong N$: 
we point out that our Corollary~\ref{cor:observational} is not obvious since  $\lambda_v^\sigma$ equates more than Plotkin's $\lambda_v$
(indeed, $=_{\beta_v} \, \subseteq \, =_\vg$ since $\to_{\beta_v} \, \subseteq \, \to_\vg$, and Example~\ref{ex:reductions} shows that this inclusion is strict).
Corollary~\ref{cor:observational} means that $\lambda_v^\sigma$ is sound with respect to the operational semantics of $\lambda_v$. 
In a way, adding $\sigma$-reduction rules to $\beta_v$-reduction is harmless with respect to Plotkin's notion of observational equivalence for $\lambda_v$: $\lambda_v^\sigma$ does not equate too much.

The converse of Corollary~\ref{cor:observational} does not hold since $\lambda x. x (\lambda y. xy) \cong \Delta$ but $\lambda x. x (\lambda y. xy)$ and $\Delta$ are different $\vg$-normal forms, so $\lambda x. x (\lambda y. xy) \neq_\vg \Delta$ 
by confluence of $\to_\vg$ (Proposition~\ref{prop:general-properties}).

\medskip

%
%
%
Another remarkable consequence of Corollary~\ref{cor:value}.\ref{cor:value.one} is Theorem~\ref{thm:valsolv} below: the notions of potential valuability and solvability for the shuffling calculus $\lambda_v^{\sigma}$ (studied in \cite{carraro14lncs}) 
coincide with the 
corresponding ones for Plotkin's $\lambda_v$ 
(studied in \cite{PaoliniRonchi99,ronchi04book,PaoliniPimentelRonchi05,paolini11tcs}). 

\begin{defi}[Potential valuability, solvability]\label{def:valsolv}
  Let $N$ be a term and $x_1, \dots, x_k$ be pairwise distinct variables (with $k \in \Nat$) such that $\Fv{N} = \{x_1, \dots, x_k\}$:
  \begin{itemize}
    \item \label{def:valsolv.potval}
    $N$ is \emph{$\vg$-potentially valuable} (resp.~\emph{$\beta_v$\!-potentially valuable}) if there are 
    values $V_1, \dots, \!V_k,\allowbreak  V
    $ such that $N\msub{V_1}{x_1}{V_k} {x_k}  \to_{\vg}^* V$ (resp.~$N\msub{V_1}{x_1}{V_k} {x_k}  \allowbreak\to_{\beta_v}^* V$);
    \item  \label{def:valsolv.solv}
    $N$ is \emph{$\vg$-solvable} (resp.~\emph{$\beta_v$-solvable}\!) whenever there are $n \in \Nat$ and terms $M_1,\dots,M_n$ 
    such that $(\lambda x_1\ldots x_k.N)M_1\cdots M_n \to_{\vg}^* I$ (resp.\ $(\lambda x_1\ldots x_k.N)M_1\cdots M_n \to_{\beta_v}^* I$).
  \end{itemize}
\end{defi}

\noindent The notions of potential valuability and solvability are parametric with respect to the reduction rules, so any variant of the $\lambda$-calculus has its own notions of potential valuability and solvability: Definition~\ref{def:valsolv} introduces them for $\lambda_v$ and $\lambda_v^\sigma$.
Clearly, potential valuability is interesting only in a call-by-value setting, where a $\beta$-redex can be 
reduced only when its argument is a value: potentially valuable terms are those that, \emph{up to a suitable substitution}, can be evaluated or 
placed in argument position without 
yielding a stuck $\beta$-redex.


The relevance of $\beta$-solvability for ordinary (call-by-name) $\lambda$-calculus is clearly presented in \cite{barendregt84nh}, where this notion has been proved to grasp the idea of ``meaningful program'', i.e., a program that can produce any given output when supplied by suitable arguments. 
It is well known that, in $\lambda$, $\beta$-solvability is operationally characterized by head $\beta$-reduction: a term is $\beta$-solvable iff it is head $\beta$-normalizable.
In a call-by-value setting, $\beta_v$-solvability and $\vg$-solvability are just the corresponding notions of solvability for $\lambda_v$ and $\lambda_v^\sigma$, respectively. 

In  \cite{PaoliniRonchi99,ronchi04book,paolini11tcs} it has been proved that $\beta_v$-solvable terms are a proper subset of the {$\beta_v$-potentially valuable} terms, and it has been pointed out that $\beta_{v}$-reduction is too weak in order to characterize both these properties: an operational characterization of {$\beta_v$-potential valuability} and $\beta_v$-solvability cannot be given inside $\lambda_v$ because of the problem of ``premature'' $\beta_v$-normal forms described in Section~\ref{sect:intro}, e.g.~the terms $M$ and $N$ in Eq.~\ref{eq:premature} are $\beta_v$-normal but neither $\beta_v$-solvable nor $\beta_v$-potentially valuable.
In fact, $\beta_v$-solvability and \mbox{$\beta_v$-potential} valuability have been operationally characterized using two lazy strategies on\,---\,\emph{call-by-name}\,---\,$\beta$-reduction (see~\cite[Theorems 3.1.9 and 3.1.14]{ronchi04book}), which is disappointing and unsound for $\lambda_v$: according to these lazy strategies, stuck $\beta$-redexes can be fired (even if the argument is not a value), for instance $(\lambda y.M)(xI)$ reduces to $M \sub{xI}{y}$.
 
On the other hand, concerning $\lambda_v^\sigma$,
Theorems 24-25 in \cite{carraro14lncs} give semantic and opera\-tional characterizations of $\vg$-potentially valuability and $\vg$-solvability.
Interestingly, the operational characterizations rest on $\vg$-reduction strategies and then are \emph{internal} to $\lambda_v^\sigma$.
Let us recall these theorems (see Proposition \ref{prop:shuffling-characterization} below) and, firstly, the notions involved in it. 

For every term $M$ with $\Fv{M} \subseteq \{x_1, \dots, x_n\}$ and $\vec{x} = (x_1, \dots, x_n)$, we denote by $\Sem{M}{\vec{x}}$ (resp.~$\Sem{M}{\vec{x}}^\Strat$) its \emph{semantics} (resp.~\emph{stratified semantics}) in a relational 
model for $\lambda_v^\sigma$ and $\lambda_v$. 
All the details about this denotational model 
are in \cite{carraro14lncs}, for our purpose it is enough to recall that $\Sem{M}{\vec{x}}$ is a set such that $\Sem{M}{\vec{x}}^\Strat \subseteq \Sem{M}{\vec{x}}$, and if $M \to_\vg N$ then $\Sem{M}{\vec{x}} = \Sem{N}{\vec{x}}$.

The reductions $\ToWeak$ and $\ToStrat$ are the closures of $\mapsto_{\beta_v} \cup \mapsto_{\sigma_1} \cup \mapsto_{\sigma_3}$ under weak and stratified contexts, respectively, where weak contexts (denoted by $\mathtt{W}$) and stratified contexts (denoted by $\mathtt{S}$) are special kinds of 
contexts defined as follows (see \cite{carraro14lncs} for more details):
\begin{align*}
  \mathtt{W} &\Coloneqq \Chole{\cdot} \mid \mathtt{W}M \mid M\mathtt{W} \mid (\lambda x. \mathtt{W})M &
  \mathtt{S} &\Coloneqq \mathtt{W} \mid \lambda x. \mathtt{S} \mid \mathtt{S}M \, .
\end{align*}

Note that $\ToWeak$ and $\ToStrat$ are two (non-deterministic but confluent) sub-reductions of $\to_\vg$.

\begin{prop}[Semantic and operational characterization of $\vg$-potential valuability and $\vg$-solvability, \cite{carraro14lncs}]
\label{prop:shuffling-characterization}
  Let $M$ be a term with $\Fv{M} \subseteq \{x_1, \dots, x_n\}$ and $\vec{x} = (x_1, \dots, x_n)$.
  \begin{enumerate}
    \item\label{prop:shuffling-characterization.valuability} \emph{Semantic and operational characterization of $\vg$-potential valuability (\cite[Theorem~24]{carraro14lncs})}: 
      $M$ is $\vg$-poten\-tially valuable iff
      $\Sem{M}{\vec{x}} \neq \emptyset$ iff
      $M$ is 
      $\Weak$-normalizable iff $M$ is strongly $\Weak$-normalizable.
    \item\label{prop:shuffling-characterization.solvability} \emph{Semantic and operational characterization of $\vg$-solvability (\cite[Theorem~25]{carraro14lncs})}:
      $M$ is $\vg$-solvable iff \allowbreak
      $\Sem{M}{\vec{x}}^\Strat \neq \emptyset$ iff
      $M$ is 
      $\Strat$-normalizable iff $M$ is strongly $\Strat$-normalizable.

  \end{enumerate}
\end{prop}

\noindent Thanks to standardization for $\lambda_v^\sigma$ (actually, Corollary~\ref{cor:value}.\ref{cor:value.one}), we can prove Theorem~\ref{thm:valsolv} 
below, which reconciles the results 
about solvability and potential valuability for $\lambda_v^\sigma$ and $\lambda_v$.

\begin{thm}[Potential valuability and solvability for $\lambda_v^\sigma$ and $\lambda_v$]
\label{thm:valsolv}
  Let $M$ be a term:
  \begin{enumerate}
    \item \label{thm:valsolv.potval}
    $M$ is $\vg$-potentially valuable if and only if $M$ is $\beta_v$-potentially valuable;
    \item  \label{thm:valsolv.solv}
    $M$ is $\vg$-solvable if and only if $M$ is $\beta_v$-solvable.
  \end{enumerate}
\end{thm}

\begin{proof} In both points, 
 the implication from right to left is trivial since $\to_{\beta_v} \, \subseteq \, \to_{\vg}$. Let us prove the other direction.
 Let $\Fv{M} = \{x_1, \dots, x_m\}$ for some $m \in \Nat$.
  \begin{enumerate}
    \item Since $M$ is $\vg$-potentially valuable, there exist 
    some values $V, V_1, \dots, V_m 
    $ 
    such that $M\msub{V_1}{x_1}{V_m}{x_m} \to_{\vg}^* V$; then, by Corollary~\ref{cor:value}.\ref{cor:value.one} and because $\ToHeadBetav \, \subseteq \, \to_{\beta_v}$, 
\[M\msub{V_1}{x_1}{V_m}{x_m} \allowbreak\to_{\beta_v}^* V'\]
   for some $V' \in \Lambda_v$. So, $M$ is $\beta_v$-potentially valuable.

    \item Since $M$ is $\vg$-solvable, there exist 
    terms $N_1,\dots,N_n$ (for some $n
    \geq 0$) such that $(\lambda x_1\ldots x_m.M)N_1\cdots N_n \to_{\vg}^* I$; then, by Corollary~\ref{cor:value}.\ref{cor:value.one} and because $\ToHeadBetav \, \subseteq \, \to_{\beta_v}$, there exists $V \in \Lambda_v$ such that $(\lambda x_1\ldots x_m.M)N_1\cdots N_n \allowbreak\to_{\beta_v}^* V 
    \ToIntTrans I$ . 
    According to Lemma~\ref{rmk:preliminary}.\ref{rmk:preliminary.interExpansion}
    , $V = \lambda x.N$ for some $N \in \Lambda$ such that $N \to_\vg^* x$.
    By Corollary~\ref{cor:value}.\ref{cor:value.one}, there is $V' \in \Lambda_v$ such that $N \ToHeadBetavReflTrans V' \ToIntTrans x$, hence $V' = x$ by Lemma~\ref{rmk:preliminary}.\ref{rmk:preliminary.interExpansion} again. Since $\ToHeadBetav \, \subseteq \, \to_{\beta_v}$, $N \to_{\beta_v}^* x$ and thus $V = \lambda x.N \to_{\beta_v}^* I$, so $M$ is $\beta_v$-solvable.
    \qedhere
  \end{enumerate}
\end{proof}


\noindent According to Theorem~\ref{thm:valsolv}, 
the notions of potential valuability and solvability for $\lambda_v^{\sigma}$ \emph{coincide} with the respective ones for Plotkin's $\lambda_v$.
So, 
the semantic (via a relational model) and operational (via two sub-reductions of $\to_\vg$) characterizations of $\vg$-potential valuability and $\vg$-solvability given 
in Proposition~\ref{prop:shuffling-characterization} are also semantic and operational characterizations of $\beta_v$-potential valuability and $\beta_v$-solvability. 
The difference is that 
these notions are characterized operationally \emph{inside} $\lambda_v^\sigma$ (using call-by-value reductions), 
while it is impossible to characterize them operationally inside $\lambda_v$.
This shows how $\lambda^{\sigma}_{v}$ is a useful, conservative and ``complete'' tool for studying semantic and operational properties of Plotkin's $\lambda_{v}$.

\medskip
For the sake of completeness, we mention another conservativity result of $\lambda_v^\sigma$ with respect to $\lambda_v$, proved in \cite[Theorem 21]{guerrieri15wpte} 
thanks to our sequentialization: it shows that the notions of head reduction for $\lambda_v^\sigma$ and $\lambda_v$ are equivalent from the termination viewpoint.

\begin{prop}[Head normalization, \cite{guerrieri15wpte}]
\label{prop:headnormalization}
  Let $N \in \Lambda$. The following are equivalent:
  \begin{multicols}{2}
  \begin{enumerate}
    \item\label{prop:headnormalization.weak} $N$ is head $\vg$-normalizable;
    \item\label{prop:headnormalization.betav} $N$ is head $\beta_v$-normalizable;
    \item\label{prop:headnormalization.equivalence} $N =_\vg L$ for some head $\vg$-normal $L$;
    \item\label{prop:headnormalization.strong} $N$ is strongly head $\vg$-normalizable.
  \end{enumerate}
  \end{multicols}
\end{prop}

The equivalence \eqref{prop:headnormalization.weak}$\Leftrightarrow$\eqref{prop:headnormalization.strong} means that normalization and strong normalization are equivalent for head $\vg$-reduction (for head $\beta_v$-reduction they are trivially
equivalent since head $\beta_v$-reduction is deterministic), therefore if one is interested in studying the termination of head $\vg$-reduction, no difficulty arises from its non-determinism.
The equivalence \eqref{prop:headnormalization.strong}$\Leftrightarrow$\eqref{prop:headnormalization.betav} or \eqref{prop:headnormalization.weak}$\Leftrightarrow$\eqref{prop:headnormalization.betav} says that the 
evaluation defined for Plotkin's $\lambda_v$ (head $\beta_v$-reduction) terminates if and only if the 
evaluation defined for $\lambda_v^\sigma$ (head $\vg$-reduction) terminates: $\sigma$-rules play no role in deciding the termination of a head $\vg$-reduction sequence (in a way, this generalizes Corollary \ref{cor:value}.\ref{cor:value.two}), they can only activate hidden $\beta_v$-redexes that are not in head position.
The equivalence \eqref{prop:headnormalization.weak}$\Leftrightarrow$\eqref{prop:headnormalization.equivalence} says that head $\vg$-reduction is complete to get head $\vg$-normal forms;
in particular, this entails that every $\vg$-normalizable term is head $\vg$-normalizable.

\medskip
Standardization is related to normalization.
In \cite[Theorem~24]{guerrieri15wpte} a family of normalizing strategies for $\lambda_v^\sigma$ has been introduced: a term $M$ is $\vg$-normalizable iff $M$ $\vg$-reduces to its $\vg$-normal form 
selecting $\vg$-redexes in a particular order defined
in \cite[Definition~22]{guerrieri15wpte}. 
Actually, these normalizing strategies 
are a special case of standard sequences.

\begin{defi}[Strict standard head sequence]
  A \emph{strict standard head sequence} is a finite sequence $(M_0, \dots, M_k, \dots, M_m)$ of terms
  (with $k \leq m$) such that $M_k$ is head $\beta_v$-normal, $M_m$ is head $\vg$-normal, $M_i \ToHeadBetav M_{i+1}$ for any $0 \leq i < k$, and
  $M_i \ToHeadSigma M_{i+1}$ for any $k\leq i < m$. 
\end{defi}

A \emph{strict standard sequence} is then defined by replacing the notion of standard head sequence with the notion of strict standard head sequence in Definition~\ref{def:std}.
So, normalization theorem proved in \cite[Theorem~24]{guerrieri15wpte} can be reformulated as follows
:

\begin{prop}[Normalization, \cite{guerrieri15wpte}]
  \label{prop:normalization}
  Let $M$ be a term: $M$ is $\vg$-normalizable iff there exists a strict standard sequence from $M$ to its $\vg$-normal form.
\end{prop}

The proof of the left-to-right direction of Proposition~\ref{prop:normalization} (the right-to-left one is trivial) relies on Proposition~\ref{prop:headnormalization}, see \cite{guerrieri15wpte} for details: the idea is that, given a $\vg$-normalizable (and then head $\vg$-normalizable) term $M$, one 
performs\,---\,deterministically\,---\,head $\beta_v$-reduction steps from $M$ as long as a head $\beta_v$-normal form $N$ is reached
(
according to Proposition~\ref{prop:headnormalization}, a term is head $\vg$-normalizable iff it is head $\beta_v$-normalizable); then, one 
performs head $\sigma$-reduction steps from $N$ (where head $\sigma_1$- and head $\sigma_3$-reduction steps can be performed in whatever order) as long as a head $\vg$-normal form $L$ is reached (such a $L$ always exists because $\ToHeadSigma$ is strongly normalizing and preserves $\beta_v$-normal forms); finally, one performs internal $\vg$-reduction steps starting from $L$ by iterating 
this strategy on the subterms of $L$, according to the standard left-to-right order, as long as the $\vg$-normal form of $M$ is reached.

Clearly, Theorem~\ref{thm:standardization} fails if in its statement ``standard sequence'' is replaced by ``strict standard sequence'': 
$I\Delta I \ToHeadSigmaOne (\lambda.x I)\Delta$ is a standard sequence but there is no strict standard sequence from $I\Delta I$ to $(\lambda.x I)\Delta$, since $I\Delta I \ToHeadBetav \Delta I \ToHeadBetav II \ToHeadBetav I$ 
and $I$ is (head) $\vg$-normal.
Similarly, $(\Delta\Delta)(II) \ToIntBetav (\Delta\Delta)I$ is a standard sequence but there is no strict standard sequence from $(\Delta\Delta)(II) $ to $ (\Delta\Delta)I$, since $(\Delta\Delta)(II)$ is not head $\beta_v$-normalizable.

\section{Conclusions}\label{sect:conclusions}

It has been proved in \cite{PaoliniRonchi99,paolini02ictcs,ronchi04book,paolini11tcs} that $\beta_{v}$-reduction is too weak to characterize operationally some semantical properties of $\lambda_{v}$, such as separability, potentially valuability and solvability%
. The main motivation behind the introduction of $\lambda_v^\sigma$ in \cite{carraro14lncs} was 
to achieve a call-by-value language where potential valuability and solvability can be characterized operationally without resorting to reductions external the call-by-value paradigm: $\lambda_{v}^\sigma$ allows an internal operational characterization of such notions \cite[Theorems 24-25]{carraro14lncs}. 
In this paper we close the game, by proving that $\lambda_{v}^\sigma$ is a conservative extension of $\lambda_{v}$: in particular, $\lambda_v^\sigma$ is sound with respect to the operational semantics of $\lambda_v$ (Corollary~\ref{cor:observational}), and potential valuability and solvability for $\lambda_v^\sigma$ coincide with the respective notions for $\lambda_v$ (Theorem~\ref{thm:valsolv}).
So, $\lambda_v^\sigma$ is a useful framework for studying semantic and operational properties of $\lambda_v$. 
The technical tool on which the proofs of these conservativity properties are based is an interesting result in its own, namely standardization for $\lambda_{v}^\sigma$ (Theorem~\ref{thm:standardization}).

Standardization for $\lambda_v^\sigma$ has been proved 
using parallel reduction. Let us recall that parallel reduction in $\lambda$-calculus has been defined by Tait and Martin-L\"of in order to prove confluence of $\beta$-reduction, without referring to the tricky notion of residuals. 
Takahashi in \cite{takahashi1989jsl,Takahashi95} has simplified this technique and showed that it can be successfully applied also to prove standardization for $\lambda$.
However, in $\lambda_v^\sigma$ our parallel reduction $\Rightarrow$ cannot be used to prove confluence of $\to_\vg$, since $\Rightarrow$ does not enjoy the diamond property.
Indeed, consider
\begin{equation*}
\scalebox{.9}{\xymatrix@C=+0cm@R=+0.2cm@M=1mm@L=.8mm{
      & (\lambda x.M) \bigl((\lambda y.N) (zz)\bigr) L \ar@{=>}[dl]_->{\text{(by applying the rule } \sigma_1\text{)} \qquad\qquad\quad}
      \ar@{=>}[dr]^->{\text{ \qquad\qquad\qquad(by applying the rule } \sigma_3\text{)}}
      \\
      M_1 = (\lambda x.ML) \bigl((\lambda y.N) (zz)\bigr)  & &  \bigl(\lambda y.(\lambda x.M)N\bigr) (zz)L = M_2 \\
    }}
\end{equation*}
  It is easy to check that there is no term $M'$ such that $M_1 \Rightarrow M'$ and 
$M_2 \Rightarrow M'$. 

The proof  of the standardization theorem is based on a sequentialization property, imposing a total order between $\beta_{v}$-redexes, but a partial one between $\sigma$-redexes. 
We conjecture that a total order between all $\vg$-redexes can be provided by defining a suitable notion of head $\sigma$-reduction that properly interleaves head $\sigma_{1}$- and head $\sigma_{3}$-reduction steps. 
Anyway, we do not fully explored this possibility because we are unaware of interesting applications.

Postponements of head $\sigma$-reduction to head $\beta_v$-reduction (Lemma~\ref{lemma:commutation}) and of internal $\vg$-reduction to head $\vg$-reduction (Corollary~\ref{cor:postponeinternal}) suggest the idea 
that, in order to avoid 
the issues affecting $\lambda_v$ when dealing with open terms and stuck $\beta$-redexes,
it is 
enough to restrict our shuffling calculus $\lambda_v^\sigma$ by allowing 
(local head) $\sigma$-reduction steps only when a 
(local head) $\beta_v$-normal form is reached.
This approach generalizes the idea behind 
strict standard sequences defined in Section~\ref{sect:conservative}.
In fact, this restricted shuffling calculus is a ``minimalistic'' extension of Plotkin's $\lambda_v$ solving the problem of premature $\beta_v$-normal forms. 
Since values are head $\beta_v$-normal and $I$ is $\vg$-normal, Corollary~\ref{cor:value}.\ref{cor:value.one} and Proposition~\ref{prop:normalization} ensure that the conservativity result given by Theorem~\ref{thm:valsolv} (as well as Corollary~\ref{cor:observational}) 
would still hold in this restricted shuffling calculus.
But solving the problem of premature $\beta_v$-normal forms is only the first step in the direction of a deep analysis of $\lambda_{v}$ and, more generally, of call-by-value settings: the \emph{whole} shuffling calculus $\lambda_v^\sigma$ seems to be an adequate framework for this task (Corollary~\ref{cor:observational}  and Theorem~\ref{thm:valsolv} exemplify how call-by-value properties can be correctly studied inside the whole $\lambda_v^\sigma$) and its study is more elegant and simpler without imposing any ``clumsy'' syntactic restrictions on the definition of shuffling calculus reduction rules.

\subsection*{Future work.} We plan to continue to explore the call-by-value setting, using the shuffling calculus $\lambda_{v}^\sigma$.
As a first step, we would like to revisit and improve the Separability Theorem given in \cite{paolini02ictcs} for $\lambda_v$.
Still the issue is more complex than in the call-by-name, indeed in ordinary $\lambda$-calculus different $\beta\eta$-normal forms can be separated (by the B\"ohm Theorem),
while in $\lambda_{v}$ there are different normal forms that cannot be separated, but which are only semi-separable (e.g. $I$ and $\lambda z.(\lambda u.z)(zz)$).
We hope to completely characterize separable and semi-separable normal forms in $\lambda_{v}^\sigma$.
This should be a first step aimed to define a semantically meaningful notion of approximants.
Then, we should be able to provide a new insight on the denotational analysis of the call-by-value, maybe overcoming limitations as that of the absence of fully abstract filter models \cite[Theorem~12.1.25]{ronchi04book}.
Last but not least, an unexplored but challenging research direction is the use of our commutation $\sigma$-rules to improve and speed up the call-by-value evaluation.
We do not have any concrete evidence supporting such possibility, but since $\lambda_{v}^\sigma$ is strongly related to the calculi presented in \cite{herbelin09lncs,accattoli12lncs} (see \cite{AccattoliGuerrieri16} for a comparison), which are endowed with explicit substitutions, we 
believe that a sharp use of commutations could have a relevant impact on the evaluation.

%

\subsection*{Acknowledgements}
The authors wish to thank the anonymous referees for their insightful comments.



\appendix
\bibliographystyle{alpha}%
\bibliography{biblio}	

\begin{thebibliography}{MOTW99}

\bibitem[Acc15]{Accattoli15}
Beniamino Accattoli.
\newblock {Proof nets and the call-by-value {$\lambda$}-calculus}.
\newblock {\em Theoretical Compuer Science}, 606:2--24, 2015.

\bibitem[AG16]{AccattoliGuerrieri16}
Beniamino Accattoli and Giulio Guerrieri.
\newblock {Open Call-by-Value}.
\newblock In {\em {Programming Languages and Systems -- 14th Asian Symposium
  (APLAS 2016)}}, volume 10017 of {\em {Lecture Notes in Computer Science}},
  pages 206--226. Springer-Verlag, 2016.

\bibitem[AP12]{accattoli12lncs}
Beniamino Accattoli and Luca Paolini.
\newblock {Call-by-Value Solvability, Revisited}.
\newblock In {\em {Functional and Logic Programming}}, volume 7294 of {\em
  {Lecture Notes in Computer Science}}, pages 4--16. Springer-Verlag, 2012.

\bibitem[AS15]{AccattoliSacerdoti15}
Beniamino Accattoli and Claudio {Sacerdoti Coen}.
\newblock {On the Relative Usefulness of Fireballs}.
\newblock In {\em {30th Annual {ACM/IEEE} Symposium on Logic in Computer
  Science, {LICS} 2015}}, pages 141--155. {IEEE Computer Society}, 2015.

\bibitem[Bar84]{barendregt84nh}
Henk Barendregt.
\newblock {\em {The Lambda Calculus: Its Syntax and Semantics}}, volume 103 of
  {\em {Studies in logic and the foundation of mathematics}}.
\newblock North Holland, 1984.

\bibitem[CF58]{CurryFeys58}
Haskell~B. Curry and Robert Feys.
\newblock {\em {Combinatory Logic}}, volume~1.
\newblock North Holland, 1958.

\bibitem[CG14]{carraro14lncs}
Alberto Carraro and Giulio Guerrieri.
\newblock {A Semantical and Operational Account of Call-by-Value Solvability}.
\newblock In {\em {Foundations of Software Science and Computation
  Structures}}, volume 8412 of {\em {Lecture Notes in Computer Science}}, pages
  103--118. Springer-Verlag, 2014.

\bibitem[CH00]{DBLP:conf/icfp/CurienH00}
Pierre{-}Louis Curien and Hugo Herbelin.
\newblock {The duality of computation}.
\newblock In {\em {Proceedings of the Fifth {ACM} {SIGPLAN} International
  Conference on Functional Programming {(ICFP} '00)}}, pages 233--243. {ACM},
  2000.

\bibitem[Cra09]{Crary09}
Karl Crary.
\newblock {A Simple Proof of Call-by-Value Standardization}.
\newblock Technical Report CMU-CS-09-137, Carnegie Mellon University, 2009.

\bibitem[DL07]{DyckhoffLengrand07}
Roy Dyckhoff and St{\'e}phane Lengrand.
\newblock {Call-by-Value lambda-calculus and {LJQ}}.
\newblock {\em Journal of Logic and Computation}, 17(6):1109--1134, 2007.

\bibitem[EG16]{EhrhardGuerrieri16}
Thomas Ehrhard and Giulio Guerrieri.
\newblock {The Bang Calculus: an untyped lambda-calculus generalizing
  call-by-name and call-by-value}.
\newblock In {\em {Proceedings of the 18th International Symposium on
  Principles and Practice of Declarative Programming (PPDP 2016)}}, pages
  174--187. {ACM}, 2016.

\bibitem[EHR92]{EgidiHonsellRonchi92}
Lavinia Egidi, Furio Honsell, and Simona {Ronchi Della Rocca}.
\newblock {Operational, Denotational and Logical Descriptions: A Case Study}.
\newblock {\em Fundamenta Informaticae}, 16(2):149--169, 1992.

\bibitem[Gir87]{Girard87}
Jean-Yves Girard.
\newblock {Linear logic}.
\newblock {\em Theoretical Computer Science}, 50(1):1--102, 1987.

\bibitem[GL02]{DBLP:conf/icfp/GregoireL02}
Benjamin Gr{\'e}goire and Xavier Leroy.
\newblock {A compiled implementation of strong reduction}.
\newblock In {\em {Proceedings of the Seventh {ACM} {SIGPLAN} International
  Conference on Functional Programming {(ICFP} '02)}}, pages 235--246. {ACM},
  2002.

\bibitem[GPR15]{guerrieri2015lipics}
Giulio Guerrieri, Luca Paolini, and Simona Ronchi~Della Rocca.
\newblock {Standardization of a Call-By-Value Lambda-Calculus}.
\newblock In {\em {13th International Conference on Typed Lambda Calculi and
  Applications (TLCA 2015)}}, volume~38 of {\em {Leibniz International
  Proceedings in Informatics (LIPIcs)}}, pages 211--225, 2015.

\bibitem[Gue15]{guerrieri15wpte}
Giulio Guerrieri.
\newblock {Head reduction and normalization in a call-by-value
  lambda-calculus}.
\newblock In {\em {2nd International Workshop on Rewriting Techniques for
  Program Transformations and Evaluation (WPTE 2015)}}, volume~46 of {\em
  {OpenAccess Series in Informatics (OASIcs)}}, pages 3--17, 2015.

\bibitem[Hin78]{hindley78}
Roger Hindley.
\newblock {Standard and normal reductions}.
\newblock {\em Transactions of the American Mathematical Society}, pages
  253--271, 1978.

\bibitem[Hue80]{huet1980acm}
G{\'e}rard Huet.
\newblock {Confluent Reductions: Abstract Properties and Applications to Term
  Rewriting Systems}.
\newblock {\em Journal of ACM}, 27(4):797--821, 1980.

\bibitem[HZ09]{herbelin09lncs}
Hugo Herbelin and St{\'e}phane Zimmermann.
\newblock {An Operational Account of Call-by-Value Minimal and Classical
  lambda-Calculus in "Natural Deduction" Form}.
\newblock In {\em {Typed Lambda Calculi and Applications, 9th International
  Conference, {TLCA} 2009}}, volume 5608 of {\em {Lecture Notes in Computer
  Science}}, pages 142--156. Springer-Verlag, 2009.

\bibitem[JGS93]{JonesGomardSestoft93}
Neil~D. Jones, Carsten~K. Gomard, and Peter Sestoft.
\newblock {\em {Partial Evaluation and Automatic Program Generation}}.
\newblock Prentice-Hall, Inc., Upper Saddle River, NJ, USA, 1993.

\bibitem[Klo80]{klopthesis}
Jan~Willem Klop.
\newblock {Combinatory Reduction Systems}.
\newblock {\em Mathematical Centre Tracts}, 127, 1980.

\bibitem[Kri90]{lcalcul}
Jean-Louis Krivine.
\newblock {\em {Lambda-Calcul~: Types et Mod{\`e}les}}.
\newblock {{\'E}tudes et Recherches en Informatique}. Masson, 1990.

\bibitem[Las05]{Lassen05}
S{\o}ren~B. Lassen.
\newblock {Eager Normal Form Bisimulation}.
\newblock In {\em {20th {IEEE} Symposium on Logic in Computer Science {(LICS}
  2005), Proceedings}}, pages 345--354. {IEEE} Computer Society, 2005.

\bibitem[Mit79]{Mitschke79}
Gerd Mitschke.
\newblock {The Standardization Theorem for $\lambda$-Calculus}.
\newblock {\em Mathematical Logic Quarterly}, 25(1-2):29--31, 1979.

\bibitem[Mog88]{moggi88ecs}
Eugenio Moggi.
\newblock {Computational lambda-calculus and monads}.
\newblock Technical report, Edinburgh University, 1988.
\newblock Tech. Report ECS-LFCS-88-66.

\bibitem[Mog89]{Moggi89}
Eugenio Moggi.
\newblock {Computational Lambda-Calculus and Monads}.
\newblock In {\em {Proceedings of the 4th Symposium on Logic in Computer
  Science {(LICS}'89)}}, pages 14--23. {IEEE} Computer Society, 1989.

\bibitem[MOTW95]{DBLP:journals/entcs/MaraistOTW95}
John Maraist, Martin Odersky, David~N. Turner, and Philip Wadler.
\newblock {Call-by-name, call-by-value, call-by-need and the linear lambda
  calculus}.
\newblock {\em Electronic Notes in Theoretical Computer Science}, 1:370--392,
  1995.

\bibitem[MOTW99]{Maraistetal99}
John Maraist, Martin Odersky, David~N. Turner, and Philip Wadler.
\newblock {Call-by-name, call-by-value, call-by-need and the linear lambda
  calculus}.
\newblock {\em Theoretical Computer Science}, 228(1--2):175--210, 1999.

\bibitem[Pao02]{paolini02ictcs}
Luca Paolini.
\newblock {Call-by-Value Separability and Computability}.
\newblock In {\em {Italian Conference in Theoretical Computer Science}}, volume
  2202 of {\em {Lecture Notes in Computer Science}}, pages 74--89.
  Springer-Verlag, 2002.

\bibitem[Plo75]{Plotkin75}
Gordon~D. Plotkin.
\newblock {Call-by-name, call-by-value and the lambda-calculus}.
\newblock {\em Theoretical Computer Science}, 1(2):125--159, 1975.

\bibitem[PPR05]{PaoliniPimentelRonchi05}
Luca Paolini, Elaine Pimentel, and Simona {Ronchi Della Rocca}.
\newblock {Lazy strong normalization}.
\newblock In {\em {Proceedings of Intersection Types and Related Systems
  (ITRS'04)}}, volume 136C of {\em {Electronic Notes in Theoretical Computer
  Science}}, pages 103--116, 2005.

\bibitem[PPR11]{paolini11tcs}
Luca Paolini, Elaine Pimentel, and Simona {Ronchi Della Rocca}.
\newblock {Strong Normalization from an unusual point of view}.
\newblock {\em Theoretical Computer Science}, 412(20):1903--1915, 2011.

\bibitem[PR99]{PaoliniRonchi99}
Luca Paolini and Simona {Ronchi Della Rocca}.
\newblock {Call-by-value Solvability}.
\newblock {\em Theoretical Informatics and Applications}, 33(6):507--534, 1999.
\newblock RAIRO Series, EDP-Sciences.

\bibitem[PR04]{paolini04iandc}
Luca Paolini and Simona {Ronchi Della Rocca}.
\newblock {Parametric parameter passing lambda-calculus}.
\newblock {\em Information and Computation}, 189(1):87--106, 2004.

\bibitem[Reg92]{Regnier92}
Laurent Regnier.
\newblock {\em {Lambda calcul et r{\'e}seaux}}.
\newblock PhD thesis, Universit{\'e} Paris 7, 1992.

\bibitem[Reg94]{Regnier94}
Laurent Regnier.
\newblock {Une {\'e}quivalence sur les lambda-termes}.
\newblock {\em Theoretical Computer Science}, 126(2):281--292, 1994.

\bibitem[RP04]{ronchi04book}
Simona {Ronchi Della Rocca} and Luca Paolini.
\newblock {\em {The Parametric $\lambda$-Calculus: a Metamodel for
  Computation}}.
\newblock {Texts in Theoretical Computer Science: An EATCS Series}.
  Springer-Verlag, 2004.

\bibitem[SF92]{sabry92lisp}
Amr Sabry and Matthias Felleisen.
\newblock {Reasoning About Programs in Continuation-passing Style.}
\newblock {\em SIGPLAN Lisp Pointers}, V(1):288--298, 1992.

\bibitem[SF93]{SabryFelleisen93}
Amr Sabry and Matthias Felleisen.
\newblock {Reasoning about programs in continuation-passing style}.
\newblock {\em Lisp and Symbolic Computation}, 6(3-4):289--360, 1993.

\bibitem[SW97]{SabryWadler97}
Amr Sabry and Philip Wadler.
\newblock {A Reflection on Call-by-Value}.
\newblock {\em {ACM} Transactions on Programming Languages and Systems},
  19(6):916--941, 1997.

\bibitem[Tak89]{takahashi1989jsl}
Masako Takahashi.
\newblock {Parallel Reduction in $ \lambda $-Calculus}.
\newblock {\em Journal of Symbolic Computation}, 7(2):113--123, 1989.

\bibitem[Tak95]{Takahashi95}
Masako Takahashi.
\newblock {Parallel Reductions in lambda-Calculus}.
\newblock {\em Information and Computation}, 118(1):120--127, 1995.

\bibitem[Ter03]{terese2003book}
Terese.
\newblock {\em {Term Rewriting Systems}}.
\newblock {Cambridge Tracts in Theoretical Computer Science}. Cambridge
  University Press, 2003.

\end{thebibliography}





\end{document}